\documentclass[acmsmall]{acmart}

\usepackage{amsmath}
\usepackage{amsfonts}
\usepackage{amsthm}
\usepackage{quiver}
\usepackage{thmtools}
\usepackage{mathtools}
\usepackage{mathpartir}
\usepackage{stmaryrd}
\usepackage{xcolor}
\usepackage{bm}
\usepackage{natbib}
\usepackage{multicol}
\usepackage{tikz-cd}

\usepackage{booktabs}   %% For formal tables:
\usepackage{subcaption} %% For complex figures with subfigures/subcaptions
\usepackage{hyperref}
\usepackage[capitalise, noabbrev]{cleveref}

\declaretheorem[name=Theorem, parent=section]{theorem}

\declaretheorem[name=Definition, parent=section, style=definition]{definition}

\allowdisplaybreaks

\newif\ifdraft
\newif\ifextended
\extendedtrue
\newcommand{\cbpv}{$\textrm{CBPV}_\omega^{\forall,\exists,\nu}$}
\newcommand{\code}[1]{\ensuremath{\texttt{#1}}} % font for all code
\newcommand{\name}[1]{\textsf{#1}}
\newcommand{\keyword}[1]{\textsf{#1}}

\newcommand{\brac}[1]{\ensuremath{\llbracket~#1~\rrbracket}}

\newcommand{\floor}[1]{\ensuremath{\lfloor~#1~\rfloor}}
\newcommand{\subst}[3]{\ensuremath{#1[#2 / #3]}}
\newcommand{\multi}[1]{\ensuremath{\{#1\}}}

\DeclareMathOperator{\Set}{\mathnormal{Set}}
\DeclareMathOperator{\Alg}{\mathnormal{Alg}}

\newcommand{\judgment}[3]{\inferrule{#1}{#2}~~{\textsc{\small [#3]}}}
\newcommand{\judgbox}[1]{\noindent \fbox{$#1$}}

\newcommand{\lam}{\ensuremath{\lambda}}
\newcommand{\with}{\&}

\newcommand{\defeq}{\triangleq}
\newcommand{\id}{\textrm{id}}

\newcommand{\el}{\textrm{el}}

\newcommand{\VTy}{\ensuremath{\textsf{VTy}}}
\newcommand{\CTy}{\ensuremath{\textsf{CTy}}}

\newcommand{\SLam}[2]{\ensuremath{\lam~#1.~#2}}
\newcommand{\AInt}{\ensuremath{\code{Int}}}
\newcommand{\AString}{\ensuremath{\code{String}}}
\newcommand{\AThk}{\ensuremath{\code{Thk}}}
\newcommand{\AThunk}[1]{\ensuremath{\AThk~#1}}
\newcommand{\AUnit}{\ensuremath{\code{Unit}}}
\newcommand{\ABool}{\ensuremath{\code{Bool}}}
\newcommand{\AProd}[2]{\ensuremath{#1 \times #2}}
\newcommand{\ACoProd}[2]{\ensuremath{+\,\{~#1~\}_{#2}}}
\newcommand{\ACoProdLit}[1]{\ensuremath{+\,\{~#1~\}}}
\newcommand{\AExists}[2]{\ensuremath{\exists~#1.~#2}}
\newcommand{\BOS}{\ensuremath{\code{OS}}}
\newcommand{\BRet}{\ensuremath{\code{Ret}}}
\newcommand{\BReturn}[1]{\ensuremath{\BRet~#1}}
\newcommand{\BArrow}[2]{\ensuremath{#1 \to #2}}
\newcommand{\BTop}{\ensuremath{\&_\varnothing}}
\newcommand{\BWith}[2]{\ensuremath{\&\{~#1~\}_{#2}}}
\newcommand{\BWithLit}[1]{\ensuremath{\&\{~#1~\}}}
\newcommand{\BNu}[2]{\ensuremath{\nu~#1.~#2}}
\newcommand{\BForall}[2]{\ensuremath{\forall~#1.~#2}}
\newcommand{\BMo}{\ensuremath{\code{RelMonad}}}
\newcommand{\BMonad}[1]{\ensuremath{\BMo~#1}}
\newcommand{\BAlg}{\ensuremath{\code{Algebra}}}
\newcommand{\BAlgebra}[2]{\ensuremath{\BAlg~#1~#2}}
\newcommand{\BMoTrans}{\ensuremath{\code{RelMonadTrans}}}

\newcommand{\Inj}[1]{\ensuremath{\keyword{inj}_{#1}}}

\newcommand{\VThunk}[1]{\ensuremath{\{#1\}}}
\newcommand{\VUnit}{\ensuremath{()}}
\newcommand{\VPair}[2]{\ensuremath{(#1, #2)}}
\newcommand{\VPack}[2]{\ensuremath{(#1, #2)}}
\newcommand{\VInj}[2]{\ensuremath{\Inj{#1}(#2)}}
\newcommand{\VHalt}{\ensuremath{\code{halt}}}

\newcommand{\MReturn}[1]{\ensuremath{\keyword{ret}~#1}}
\newcommand{\MTmLam}[2]{\ensuremath{\lambda~#1.~#2}}
\newcommand{\MTyLam}[2]{\ensuremath{\Lambda~#1.~#2}}
\newcommand{\MBind}[3]{\ensuremath{\keyword{do}~#1 \gets #2~\keyword{;}~#3}}
\newcommand{\MForce}[1]{\ensuremath{~!~#1}}

\newcommand{\MLet}[3]{\ensuremath{\keyword{let}~ #1 = #2 ~\keyword{in}~ #3}}

\newcommand{\MLetPair}[4]{\ensuremath{\keyword{let}~\VPair{#1}{#2} = #3 ~\keyword{in}~ #4}}
\newcommand{\MMatch}[4]{\ensuremath{\keyword{match}~#1~\{~#2 \Rightarrow #3~\}_{#4}}}
\newcommand{\MLetPack}[4]{\ensuremath{\keyword{let}~\VPack{#1}{#2} = #3 ~\keyword{in}~ #4}}
\newcommand{\MTop}{\ensuremath{\keyword{comatch}~|}}
\newcommand{\MComatch}[3]{\ensuremath{\keyword{comatch}~\{~#1 \Rightarrow #2~\}_{#3}}}
\newcommand{\MRoll}[1]{\ensuremath{\keyword{roll}(#1)}}
\newcommand{\MUnroll}[1]{\ensuremath{\keyword{unroll}(#1)}}
\newcommand{\MFix}[2]{\ensuremath{\keyword{fix}~#1.~#2}}
\newcommand{\MMonadic}[1]{\ensuremath{\mathcal{M}(#1)}}

\newcommand{\Implicit}[1]{\ensuremath{\{#1\}}}

\newcommand{\CtxtEmpty}{\ensuremath{~\cdot~}}
\newcommand{\CtxtExtend}[3]{\ensuremath{#1,~#2:#3}}

\newcommand{\Stk}{\ensuremath{\mathcal{K}}}
\newcommand{\SAmbient}{\ensuremath{~\bullet~}}
\newcommand{\SKont}[3]{\ensuremath{\textsf{Kont}(#1~.~#2) :: #3}}
\newcommand{\SApp}[2]{\ensuremath{\textsf{App}(#1) :: #2}}
\newcommand{\STyApp}[2]{\ensuremath{\textsf{TyApp}(#1) :: #2}}
\newcommand{\SDtor}[2]{\ensuremath{\textsf{Dtor}(#1) :: #2}}
\newcommand{\SUnroll}[1]{\ensuremath{\textsf{Unroll} :: #1}}

\newcommand{\SigD}[1]{\ensuremath{\textrm{Sig}(#1)}}
\newcommand{\SigP}[2]{\ensuremath{\textrm{Sig}_{\,#1}(#2)}}
\newcommand{\StrP}[1]{\ensuremath{\textrm{Str}(#1)}}

\newcommand{\flCar}[1]{\ensuremath{\floor{#1}}}
\newcommand{\AMV}{\ensuremath{A}}
\newcommand{\BMV}{\ensuremath{B}}

\newcommand{\ctxAssignKind}[3]{\ensuremath{#1 \vdash #2 : #3}}
\newcommand{\ctxAssignType}[3]{\ensuremath{#1 \vdash #2 : #3}}

\acmJournal{PACMPL}
\acmVolume{9}
\acmNumber{OOPSLA1} % CONF = POPL or ICFP or OOPSLA
\acmArticle{100}
\acmYear{2025}
\acmMonth{4}
\acmDOI{10.1145/3720434} % \acmDOI{10.1145/nnnnnnn.nnnnnnn}
\startPage{1}

\ifextended
\setcopyright{none}
\else
\setcopyright{cc}
\setcctype{by}
\fi
\begin{document}

%% Title information
\title{Notions of Stack-Manipulating Computation and Relative~Monads}
  %% \title{Relative Monads for Stack-Based Functional Programming}         %% [Short Title] is optional;
  %% \title{Stack-based Effect implementations as Relative Monads}                                        %% when present, will be used in
                                          %% header instead of Full Title.
  %% \titlenote{with title note}             %% \titlenote is optional;
                                          %% can be repeated if necessary;
                                          %% contents suppressed with 'anonymous'
\ifextended
\subtitle{Extended Version}
\subtitlenote{This is an extended version of a paper to appear at OOPSLA 2025, containing appendices with additional technical details and examples.}
\fi

%% \subtitle{Subtitle}                     %% \subtitle is optional
  %% \subtitlenote{with subtitle note}       %% \subtitlenote is optional;
                                          %% can be repeated if necessary;
                                          %% contents suppressed with 'anonymous'

%% Author information
%% Contents and number of authors suppressed with 'anonymous'.
%% Each author should be introduced by \author, followed by
%% \authornote (optional), \orcid (optional), \affiliation, and
%% \email.
%% An author may have multiple affiliations and/or emails; repeat the
%% appropriate command.
%% Many elements are not rendered, but should be provided for metadata
%% extraction tools.

%% Author with single affiliation.
\author{Yuchen Jiang}
%% \authornote{with author1 note}          %% \authornote is optional;
                                        %% can be repeated if necessary
\orcid{0000-0002-1274-1450}             %% \orcid is optional
\affiliation{
%%   \position{}
  \department{Computer Science and Engineering}              %% \department is recommended
  \institution{University of Michigan}            %% \institution is required
  % \streetaddress{2260 Hayward St.}
  \city{Ann Arbor}
  \state{MI}
  % \postcode{48109-2121}
  \country{USA}                    %% \country is recommended
}
\email{lighght@umich.edu}          %% \email is recommended

%% Author with single affiliation.
\author{Runze Xue}
%% \authornote{with author1 note}          %% \authornote is optional;
                                        %% can be repeated if necessary
\orcid{0000-0003-1274-0922}             %% \orcid is optional
\affiliation{
%%   \position{}
  \department{Computer Science and Engineering}              %% \department is recommended
  \institution{University of Michigan}            %% \institution is required
  % \streetaddress{2260 Hayward St.}
  \city{Ann Arbor}
  \state{MI}
  % \postcode{48109-2121}
  \country{USA}                    %% \country is recommended
}
\affiliation{
%%   \position{}
  \department{Department of Computer Science and Technology}              %% \department is recommended
  \institution{University of Cambridge}            %% \institution is required
  % \streetaddress{15 JJ Thomson Ave.}
  \city{Cambridge}
  \state{Cambridgeshire}
  % \postcode{CB3 0FD}
  \country{UK}                    %% \country is recommended
}
\email{rx245@cam.ac.uk}          %% \email is recommended

\author{Max S. New}
%% \authornote{with author1 note}          %% \authornote is optional;
                                        %% can be repeated if necessary
\orcid{0000-0001-8141-195X}
\affiliation{
  \position{Assistant Professor}
  \department{Computer Science and Engineering}              %% \department is recommended
  \institution{University of Michigan}            %% \institution is required
  % \streetaddress{2260 Hayward St.}
  \city{Ann Arbor}
  \state{MI}
  % \postcode{48109-2121}
  \country{USA}                    %% \country is recommended
}
\email{maxsnew@umich.edu}          %% \email is recommended

%% Abstract
%% Note: \begin{abstract}...\end{abstract} environment must come
%% before \maketitle command
\begin{abstract}
  Monads provide a simple and concise interface to user-defined
  computational effects in functional programming languages. This
  enables equational reasoning about effects, abstraction over monadic
  interfaces and the development of monad transformer stacks to
  allow for multiple effects. Compiler implementors and assembly code
  programmers similarly virtualize effects, and would benefit from
  similar abstractions if possible. However, the implementation
  details of effects seem disconnected from the high-level monad
  interface: at this lower level much of the design is in the layout
  of the runtime \emph{stack}, which is not accessible in a high-level
  programming language.

  We demonstrate that the monadic interface can be faithfully adapted
  from high-level functional programming to a lower level setting with
  explicit stack manipulation. We use a polymorphic call-by-push-value
  (CBPV) calculus as a setting that captures the essence of
  stack-manipulation, with a type system that allows programs to
  define domain-specific stack structures. Within this setting, we
  show that the existing category-theoretic notion of a \emph{relative
  monad} can be used to model the stack-based implementation of
  computational effects. To demonstrate generality, we adapt a variety
  of standard monads to relative monads. Additionally, we show that
  stack-manipulating programs can benefit from a generalization of
  do-notation we call ``monadic blocks'' that allow all CBPV code to
  be reinterpreted to work with an arbitrary relative monad. As an
  application, we show that all relative monads extend automatically
  to relative monad transformers, a process which is not automatic for
  monads in pure languages.
\end{abstract}

%% 2012 ACM Computing Classification System (CSS) concepts
%% Generate at 'http://dl.acm.org/ccs/ccs.cfm'.

%% End of generated code

%% Keywords
%% comma separated list
  %% \keywords are mandatory in final camera-ready submission

%% \maketitle
%% Note: \maketitle command must come after title commands, author
%% commands, abstract environment, Computing Classification System
%% environment and commands, and keywords command.
\maketitle

\section{Introduction}
\label{sec:introduction}

Since Moggi's seminal work \cite{moggi91}, monads have become a wildly
successful tool in both semantics of programming languages and
functional programming. Moggi showed that denotational models of
call-by-value programming languages with effects can be easily
constructed using the Kleisli category of a monad. Later, Wadler and
others popularized this construction as a programming
construct\cite{wadlermonads, haskelliomonad}. Haskell programmers
design custom monads that capture the effects in their application,
and then use do-notation as an embedded language for call-by-value
programming within the ``pure'' host language Haskell.

Using monads, a programmer can implement virtualized effects: effects
that are not built into the language implementation but instead
simulated via an encoding using host language features. For instance,
even if the host language does not include a mechanism for exceptions,
programmers can write functions that use an output type of
\,\code{Either} \code{e} \code{a}\, representing an effectful
computation that may either ``raise'' an exception of type \code{e} or
return a value of type \code{a}. If a programmer extensively uses the
monadic bind operation, then errors propagate just as if language
supported exceptions directly. Monadic programming with \code{Either}
can then be viewed as programming with an embedded language that does
have support for exceptions. Implementing the language means
implementing the \code{Either} \code{e} \code{a} type, the monadic
operations, as well as the implementation of raising errors.

Compiler writers also implement virtualized effects: exceptions and
other continuation-based effects are typically implemented in software
rather than directly using a hardware-level mechanism. How do these
low-level implementation techniques compare to the high-level version
of exceptions? At first they do not seem so similar. When a compiler
writer implements a language with exceptions, they may reach for one
of several well-known techniques. One is to use \emph{double-barreled}
continuations (folklore but formalized by \citet{double-barrel}),
where instead of storing a single continuation in the form of a return
address on the stack, two continuations are passed: one for success
and one for failure. Another is to use \emph{stack-walking}
techniques: markers are placed on the stack to indicate the presence
of particular exception \emph{handlers} and raising an exception
involves searching the stack dynamically for the closest relevant
handler. In either case, the compiler writer designs the layout of the
stack, where continuations or handlers are stored, and how returns and
raises are implemented with these stack structures.

It appears hopeless to try to fit these low-level implementation
techniques into a monadic interface: fundamentally monads are an
abstraction of \emph{higher-order} effectful programming: a value of
type \code{m} \code{a} for a monad \code{m} is a
\emph{first-class} value representing a computation that performs
effects and produces \code{a} values. This seems mostly orthogonal
to the design decisions above, effectful computations are made
first-class by using \emph{closures}, but closures never need to
explicitly come up when describing effectful computations, and
first-order programming languages can support effects without
implementing closures at all.

On the other hand, there are similarities between monadic programming
and low-level effect implementations. The monadic interface is given
by a type structure, and similarly the low-level interface needs to
describe the appropriate invariants on the machine state, e.g., where
continuations or handlers are placed on the stack. Just like the
monadic \code{return}, the low-level implementation must be able to
return a final result. Analogous to the monadic bind, the low-level
implementation must be able to extend the current stack with further
continuations in order to implement function calls.

In this work we seek to resolve the disconnect between high-level and
low-level views of effects. We develop a suitable variation of
the monadic interface that on the one hand has the modularity and
abstraction benefits of high-level monadic programming, while applying
directly to the more low-level view of effects in compiler
implementations. In particular, the low-level aspect we focus on is
the design of layouts for continuations on the runtime
\emph{stack}. To study this problem in the abstract, we swap out
Haskell as our implementation language for a polymorphic
call-by-push-value (CBPV) calculus \cite{levy01:phd,levy-stack-2005}. Just as $\lambda$-calculus is the
``standard model'' for programming with functions, we view CBPV as the
``standard'' model for programming with the stack. This means that
low-level concepts like stack frames, stack walking, argument and
continuation passing are represented explicitly as emergent
programming constructs within the CBPV calculus, rather than being
built-in notions. Most importantly, CBPV features a \emph{type system}
that allows us to describe stack structures in the same way that a
Haskell-like type system allows us to define algebraic data types.

In more detail, CBPV has two kinds of types. Value types \code{A}
behave like types in a call-by-value language, in that they classify
first-class data and include algebraic data types. Computation types
\code{B} on the other hand, classify computations that interact with
the runtime stack of the abstract machine. In particular, we can view
the type \code{B} as describing the possible states that the stack
can be in when the computation executes: how many arguments are on the
stack, if the stack stores a continuation, etc. This interpretation of CBPV is concretely realized with its stack-based
operational semantics, where type safety says that the possible
structures the stack in the abstract machine can take must conform to
the type of the current computation, just as values in the environment
must conform to the type of the variables.

We enrich CBPV with System $F_\omega$-style of polymorphism over types
and type constructors, allowing us to define Haskell-like abstractions
over value and computation types. We show that this expressive type
system can describe an analog of monads that applies directly to the
stack-based implementation strategies such as those for exceptions
described above. The key idea is that the correct type of the
``monad'' is not an endofunctor from value types to value types, but a
functor from value types to \emph{computation types}. Intuitively, our
monads take a type of value to the type of computations that can
perform the effects specified by the monad. Viewed in terms of stacks,
the monad defines the structure of \emph{continuations} implemented
using the stack that allows for not just returning but also some other
kind of side-effects. Then the analog of return and bind are not pure
functions on values, but instead stack-manipulating computations.
Our adapted notion of monad is not just some ad hoc definition, but
fits into a well-known generalization of monads: \emph{relative
monads}, which generalize monads so that they need not be
endofunctors\cite{acu-2010}. 

\paragraph{Overloading Computations using Monads}

One of the most important aspects of monads is their ability to
``overload the semicolon'' by instantiating a quite simple structure:
a type constructor, and two operations on it called return and bind,
satisfying some natural equations. This gives Haskell programmers
access to the convenience of do-notation, which can be viewed as a
shallow embedding of call-by-value programming into the pure
metalanguage.

A natural question we then turn to is what the analog of do-notation
is for relative monads in CBPV. We show that there is an even more
powerful version of the do-notation in that a relative monad allows
not just for embedding of call-by-value programming into CBPV, but in
fact all of CBPV programming! That is, the embedded language the
relative monad provides is the same as the metalanguage, meaning all
constructs in CBPV can be ``overloaded'' to work with an arbitrary
user-defined effect. This is because CBPV is already a calculus for
effectful programming, rather than pure programming. CBPV supports an
abstract computation type $\BReturn{A}$ that can be viewed as the type
of computations that can perform the ``built-in'' effects. Since this
type is abstract, all CBPV code that uses it can be freely
reinterpreted as code where $\BReturn{A}$ is replaced by an arbitrary
user-defined relative monad. This includes the stack-manipulating
computations themselves, where the interpretation of the stack now
replaces the abstract continuations for $\BReturn{A}$ with the
user-defined notion given by the relative monad.

To demonstrate this overloadability, we extend our CBPV calculus with
a new computation construct we call a \emph{monadic block}, which
re-interprets arbitrary CBPV computations as interacting with a
user-defined relative monad. This can be implemented by a
source-to-source transformation, where the abstract return and bind
provided by the $\BReturn{A}$ type are translated to the return of the
monad as well as the use of \emph{algebras} for the relative monad,
which are constructed in a type-directed manner. For this reason we
call this translation the \emph{algebra translation} as it interprets
computation types as algebras for the user-defined monad.

We apply this overloadability to study another aspect of programming
with monads: monad transformers. Monad transformers can be thought of
as a kind of type constructor for monads: they take a monad as input
and return a new monad which allows for additional effects to be
performed. This allows for monads to be constructed compositionally as
large ``monad transformer stacks''. For this reason, most monads in
Haskell are in fact defined in transformer form to enable this
compositionality. While some techniques are known for deriving a monad
transformer from a monad \cite{hinze-2000}, there is no
automatic mechanism for extending a monad to a monad transformer. Most
well-known is that the list monad only extends to a monad transformer
if the input monad is assumed to be commutative.

In studying relative monad transformers in CBPV we encounter a quite
different situation: because every monad is relative to the ambient
notion of effect encoded by $\BReturn{A}$, \emph{all} monads definable
in CBPV can be mechanically extended to monad transformers. The
mechanical process is in fact a special case of our monadic blocks: we
define the relative monad transformer from a relative monad $S$ by
taking in an input relative monad $T$ and constructing $S$ within a
monadic block for $T$.

\paragraph{Overview}

The remainder of the paper is structured as follows
\begin{itemize}
\item In \cref{sec:cbpv-intro}, we give an introduction to our
  polymorphic CBPV calculus and show examples of how computations
  should be viewed as manipulating the runtime stack. We then define
  the syntax typing and stack-based operational semantics in
  \cref{sec:zydeco-syntax-semantics}.
\item In \cref{sec:relative-monads-in-zydeco} we introduce relative
  monads as a programming abstraction in CBPV and demonstrate 
  several examples how we can abstract over uses of the stack to
  implement state and control effects.
\item In \cref{sec:mo-blocks-and-alg-trans} we extend CBPV with
  \emph{monadic blocks}, show how these can be implemented using
  \emph{algebras} of a relative monad, and apply monadic blocks to
  derive relative monad transformers automatically from relative
  monads.
\item In \cref{sec:fund-thm} we connect our CBPV programming
  constructs with their corresponding category-theoretic notions. We
  show that the implementation of monadic blocks corresponds to a
  theorem we call the \emph{fundamental theorem of CBPV relative
  monads}.
\item Finally, in \cref{sec:discussion} we discuss related and future work.
\end{itemize}

All examples in the paper have been implemented in a proof of concept
implementation of our calculus we call Zydeco which supports
polymorphic CBPV with monadic blocks \cite{zydeco}. Zydeco makes CBPV
programming more ergonomic by incorporating a bidirectional type
system as well as unifying most of the core primitive type
constructors under generalized \code{data} and \code{codata}
forms. Additionally it includes primitive types and functions to allow
for simple shell scripting examples.

\section{Programming with the Stack in Call-by-push-value}
\label{sec:cbpv-intro}

We begin with examples of programming in our calculus
\cbpv{} and how we interpret
programs as manipulating the stack.

\paragraph{Calling Conventions as Types}

CBPV is a low-level language in that there is a syntactic distinction
between \emph{values} and \emph{computations}. Levy's slogan is that
``a value is'' and ``a computation does''. Values are inert data that
are first class and can be stored in variables, returned to
continuations and passed as arguments on the stack. Computations on
the other hand are ``imperative'' in that they execute by manipulating
the current state of the stack.
Values and computations each have their own corresponding kinds of
types: value types and computation types.

As a simple example, consider a function that implements a polynomial
$x^2+x+10$. In our CBPV syntax this function could be written as a
\emph{computation} as follows:
\begin{footnotesize}
\begin{align*}
  \MTmLam{{\code{x}}}&\keyword{do}~\code{s}~\leftarrow~!~\code{times}~\code{x}~\code{x};\\*
    &\keyword{do}~\code{y}~\leftarrow~!~\code{add}~\code{x}~10;\\*
    &!~\code{add}~\code{s}~\code{y}
\end{align*}
\end{footnotesize}
The style of program resembles monadic or ANF intermediate representations:
the results of calls to functions are explicitly bound to variables using a
\keyword{do} binding. This makes the evaluation order explicit in the
program text.
Operationally, we read this function in an imperative manner.
First, the $\MTmLam{\code{{}x}}$ is a command to \emph{pop} an argument
off of the stack and store it in a local variable ${\code{x}}$. Next, we have a
\keyword{do} binding around a call to \code{times} with two
arguments. The imperative reading of this operation is that we
\emph{push} first the continuation for the ${s}$ value, and then the two
copies of the argument ${\code{x}}$ onto the stack, before forcing the
execution of the $\code{times}$ function. When that function returns,
we will execute the continuation, in which we perform a similar
process: push the continuation for the \code{add} operation, push the
argument $10$ then the argument ${\code{x}}$ and force the execution of the
$\code{add}$ function. When this function returns, we will push the
arguments ${\code{y}}$ and then ${\code{s}}$ and perform a tail call to the $\code{add}$
function.
In this way, CBPV computations provide an abstract model of
programming with a call-stack, using a stack-based calling
convention, i.e., without global registers, but with private local variables.

The typing of this function matches this operational
reading. The type of this polynomial computation is a \emph{computation type}
$\AInt \to \BReturn\AInt$, where $\AInt$ is a base value type of
fixed-precision integers. Computation types can be read in two dual ways: in terms of
how the computation behaves or in terms of the \emph{stack} that the
computation interacts with. From the computation perspective, the
function type indicates that the computation pops an argument $\AInt$
off of the stack, and then behaves as a $\BReturn \AInt$
computation. From the stack's perspective, the function type indicates
that before such a computation runs, the stack must consist of an
$\AInt$ pushed onto a stack implementing an $\AInt$ continuation. So
just as CBPV computations are an abstraction of
call-stack-manipulating programs, CBPV computation \emph{types} can be
viewed as a type structure for call stacks, and so can be used to
encode stack-based \emph{calling conventions}. The types of
$\code{add}$ and $\code{times}$ are here assumed to be
$\AThk{(\AInt\to\AInt\to\BReturn{\AInt})}$. The $\AThk{}$ here
indicates that $\code{add}$ and $\code{times}$ are \emph{thunked} or
\emph{suspsended} computations, meaning they are first-class values
that could be passed as arguments or returned from
functions. Operationally, this is implemented as a \emph{closure}.

\paragraph{Complex calling conventions}

So far we've seen two computation type constructors: $A \to B$, which
represent computations that expect an $A$ argument at the top of the
stack, and $\BReturn{A}$ which represent computations that return to
an $A$-continuation. We've also seen that these computation types can
be used compositionally: a function that expects two arguments on the
stack can be given type $A_1 \to A_2 \to B$, which is equivalent to
the uncurried $A_1\times A_2 \to B$, unlike in call-by-value
languages.

We next show how to encode more complex calling conventions: functions
with an optional argument ($\to^?$), functions that take any number
number of arguments ($\to^*$), and functions that take at least one
argument ($\to^+$):
\begin{itemize}
\item $A \to^? B$ is defined as $\BWithLit{~\code{.some}: A \to B,~\code{.none}: B~}$,
\item $A \to^* B$ is defined as $\BWithLit{~\code{.more}: A \to A \to^* B,~\code{.done}: B~}$, and
\item $A \to^+ B$ is compositionally defined as $A \to A \to^* B$.
\end{itemize}
The type $\&\{B_d\}_{d\in D}$ is a \emph{lazy product}
computation type with named projections $d \in D$. We use a linear logic
notation $\&$ for this product to distinguish it from the value
product $A \times A'$, which is a strict product whose values are
pairs.
To understand $\&$, consider the $\to^?$ type. From the computation perspective, the type $A \to^? B$ is
a lazy pair of computations: one that takes one argument $A \to B$,
accessible with the projection named $\code{.some}$ and the other,
which takes no $A$ argument, accessible with the projection named
$\code{none}$.
From the stack perspective, the stack for a lazy product is like a sum
type: the stack consists of a \emph{destructor name} pushed onto the
stack for the corresponding type. So for $A \to^? B$, the stack
consists of either a $\code{.some}$ followed by an $A$ and a stack for
$B$ or a $\code{.none}$ followed by a stack for $B$.
Then computations of product type are implemented by \emph{co}-pattern
matching on the stack to pop the destructor and behave according to
which destructor is found. For instance, a function $\AInt \to^?
\BReturn \AInt$ that either returns the input value, defaulting to $0$
if none is provided would be implemented as
\begin{footnotesize}
\[
\keyword{comatch}~|~\code{.some}~\Rightarrow~\MTmLam{{\code{n}}}{\MReturn{\code{n}}}~|~\code{.none}~\Rightarrow~\MReturn{0}
\]
\end{footnotesize}
We use a post-fix notation for destructors, making them look somewhat
like methods in an object-oriented language, this matches the syntax
for argument passing with the function type, as both are interepreted
as pushing something onto the stack: a destructor for $\&$ and a value
for $\to$. Here $\MReturn{V}$ is the syntax for returning a value $V$
to the continuation currently stored on the stack.

The next example of a variable-arity function $A \to^* B$ is notable
in that it is a \emph{recursive} computation type. A function that
takes variable number of arguments is implemented here as a product of
a function that takes one argument $\code{.none} : B$ and a function
that takes one argument and then behaves recursively as a function
taking a variable number of arguments $\code{.more} : A \to A \to^*
B$. Viewed from the perspective of the stack, this is similar to a
cons-list: the stack is either tagged as $\code{.none}$ and contains a
$B$ stack, or is tagged as $\code{.more}$ and contains an $A$ value
pushed onto an $A\to^* B$ stack.

As an example of using recursive computation types, consider
a recursive function $\code{sum\_and\_mult} : \AThunk{(\AInt \to^+
  \AInt \to^? \BReturn{\AInt})}$ that sums up one or more arguments
and, if provided, multiplies the result by an optionally provided
multiplier.
\begin{footnotesize}
\begin{align*}
  (\MFix{\code{loop} &: \AThk{(\AInt \to \AInt\to^*\AInt\to^?\BReturn{\AInt})}}\MTmLam{\code{sum}}\keyword{comatch}
    \\ &|~\code{.more}~\Rightarrow\MTmLam{{\code{i}}}~\MBind{{\code{j}}}{\MForce{\code{add}~\code{sum}~{\code{i}}}}{\MForce{\code{\code{loop}}~{\code{j}}}}
    \\ &|~\code{.done}~\Rightarrow~\keyword{comatch}
    \\ &\qquad|~\code{.some}~\Rightarrow\MTmLam{{\code{i}}}~\MForce{\code{mult}~\code{sum}~{\code{i}}}
    \\ &\qquad|~\code{.none}~\Rightarrow~\MReturn{\code{sum}})~0
\end{align*}
\end{footnotesize}
We use $\keyword{fix}$ to construct a recursive thunked computation
$\code{loop}$ that implements a tail-recursive loop. The loop takes an
extra accumulator parameter $\code{sum}$ which is initialized to $0$.
After popping the accumulator $\code{sum}$, the loop
$\keyword{comatch}$es on the stack.  If the top of the stack is a
$\code{.more}$, we pop the next argument, add it to the sum and
continue the loop.  If the top of the stack is a $\code{.done}$, we
$\keyword{comatch}$ on the remainder of the stack.  If the top of the
stack is $\code{.some}$, we pop the argument, multiply it by the sum
by the argument and return the result, whereas the top of the stack is
$\code{.none}$, we directly return the sum.

Operationally, this computation ``walks'' up a stack which may be of
arbitrary size, using the destructors as stack tags to identify where
data is stored. Later we will use a similar approach to model
stack-walking exception handling.

\paragraph{Abstract Machine Interpreter}
\label{sec:abstract-machine-cbv}

\begin{figure}
  \begin{footnotesize}
    
\begin{align*}
  \code{Expr} : \VTy~&\defeq~+\,\{
                \,\,\code{Var}: \AString,~
                    \code{Lam}: \AString \times \code{Expr},~
                    \code{App}: \code{Expr} \times \code{Expr},
    \\ & \qquad\quad\code{True}: \AUnit,~
                    \code{False}: \AUnit,~
                    \code{If}: \code{Expr} \times \code{Expr} \times \code{Expr}
  \,\,\}\\
  \code{Value} : \VTy~&\defeq~+\,\{
                \,\,\code{True}: \AUnit,~
                    \code{False}: \AUnit,~
                    \code{Closure}: \AThunk{(\code{Value} \to \code{Machine})}
  \,\,\}\\
  \code{Machine} : \CTy~&\defeq~\&\{
                \,\,\code{.fun}: \code{Expr} \to \code{Env} \to \code{Machine},~
                    \code{.app}: \code{Value} \to \code{Machine},
    \\ & \qquad\quad\code{.if}: \code{Expr} \to \code{Expr} \to \code{Env} \to \code{Machine},~
                    \code{.end}: \BReturn{\code{Answer}}
  \,\,\}\\
  \code{Answer} : \VTy~&\defeq~+\,\{
                \,\,\code{Ok}: \code{Value},~
                    \code{Err}: \AUnit
  \,\,\}
\end{align*}
\begin{align*}
  &\code{descend}
  : \AThunk{(\code{Expr} \to \code{Env} \to \code{Machine})}~\defeq~\{
    \\ &\qquad\MTmLam{\code{e}~\code{g}}{\keyword{match}~\code{e}}
    \\ &\qquad\qquad |~\,\code{Var}(\code{x}) \Rightarrow \MBind{v}{\MForce{\code{lookup}}~\code{x}~\code{g}}{}
    \\ &\qquad\qquad\qquad \code{match}~\code{v}
    \\ &\qquad\qquad\qquad |~\,\code{Err()} \Rightarrow \MForce{\code{error}}
    \\ &\qquad\qquad\qquad |~\,\code{Ok}(\code{v}) \Rightarrow \MForce{\code{ascend}}~\code{v}
    \\ &\qquad\qquad |~\,\code{True()}~ \Rightarrow \MForce{\code{ascend}}~\code{True()}
    \\ &\qquad\qquad |~\,\code{False()}~ \Rightarrow \MForce{\code{ascend}}~\code{False()}
    \\ &\qquad\qquad |~\,\code{Lam}~(\code{x},\code{e}) \Rightarrow \MForce{\code{ascend}}~\code{Closure}(\VThunk{~
      \MForce{\code{descend}}~\code{e}~((\code{x}, \code{v}) :: \code{g})
    ~})
    \\ &\qquad\qquad |~\,\code{App}~(\code{f},\code{a}) \Rightarrow \MForce{\code{descend}}~\code{f}~\code{g}~\code{.fun}~\code{a}~\code{g}
    \\ &\qquad\qquad |~\,\code{If}~(\code{c},\code{t},\code{f}) \Rightarrow \MForce{\code{descend}}~\code{c}~\code{g}~\code{.if}~\code{t}~\code{f}~\code{g}
    \,\,\}
  \\
  &\code{ascend}
  : \AThunk{(\code{Value} \to \code{Machine})}~\defeq~\{
    \\ &\qquad\MTmLam{\code{v}}{}\code{comatch}
    \\ &\qquad\qquad |~\code{.fun}\Rightarrow \MTmLam{{\code{a}}~{\code{g}}}\MForce{\code{descend}}~\code{a}~\code{g}~\code{.app}~\code{v}
    \\ &\qquad\qquad |~\code{.app}\Rightarrow \MTmLam{\code{f}}\code{match}~\code{f}
    \\ &\qquad\qquad\qquad |~\,\code{True()}~|~\,\code{False()} \Rightarrow \MForce{\code{error}}
    \\ &\qquad\qquad\qquad |~\,\code{Closure}(c) \Rightarrow !~\code{c}~\code{v}
    \\ &\qquad\qquad |~\code{.if} \Rightarrow \MTmLam{\code{t}~\code{f}~\code{g}}\code{match}~\code{v}
    \\ &\qquad\qquad\qquad |~\,\code{True()} \Rightarrow \MForce{\code{descend}}~\code{t}~\code{g}
    \\ &\qquad\qquad\qquad |~\,\code{False()} \Rightarrow \MForce{\code{descend}}~\code{f}~\code{g}
    \\ &\qquad\qquad\qquad |~\,\code{Closure}(\_) \Rightarrow \MForce{\code{error}}
    \\ &\qquad\qquad |~\code{.end} \Rightarrow \MReturn{\code{Ok}(\code{v})}
  \,\,\}
  \\
  &\code{error} : \AThunk{\code{Machine}}~\defeq~\{
    \,\,\keyword{comatch}
    \\ &\qquad\qquad|~\code{.fun}~|~\code{.app}~|~\code{.if}~\Rightarrow \MTmLam{\_} \MForce{\code{error}}
    \\ &\qquad\qquad|~\code{.end} \Rightarrow \MReturn{\code{Err()}}
  \,\,\}\\
  &\code{eval} : \AThunk{\code{Expr} \to \BReturn{\code{Answer}}} ~\defeq~\{~\MTmLam{\code{e}}~\MForce{\code{descend}}~\code{e}~\code{[]}~\code{.end}~\}
\end{align*}
  \end{footnotesize}
\caption{Abstract Machine-Based Interpreter}
\label{fig:abstract-machine-cbv}
\end{figure}

As a larger example of using CBPV as a language for
stack-manipulation, we implement a stack-based abstract machine
interpreter for an untyped call-by-value lambda calculus in
\cref{fig:abstract-machine-cbv}.
Our object language syntax is encoded in a recursive $\code{Expr}$
type implemented as a labeled sum type of variables, booleans,
if-expressions, lambdas or applications.
At runtime, the semantic values of the language are booleans or thunks
of functions that take one argument, which implement the closures in
the object language using CBPV closures.
The interesting aspect of our interpreter is that we implement the
interaction with the stack of the abstract machine using the ambient
notion of stack in CBPV. That is, we define a computation type of
$\code{Machine}$s which interact with the stack of the object
language.
The $\code{Machine}$ type is implemented as a product with the
four different destructors corresponding to the different
``frames'' of the abstract machine stack.
A $\code{.fun}$ frame indicates we are evaluating the function part of an application, so we are given an argument $\code{Expr}$ to be evaluated next in the provided $\code{Env}$. An $\code{.app}$ frame indicates we are evaluating the argument to a function application and when we are done we should apply the
provided function $\code{Value}$ to it. An $\code{.if}$ frame indicates
we are evaluating the condition of an $\code{if}$ expression, and
based on its value we will either execute either of the provided branch $\code{Expr}$s
with the provided $\code{Env}$. Finally, an $\code{.end}$ frame indicates
the base case of an ``empty'' stack, in which case we should return a final result of our interpreter: either
the value of the expression that is being evaluated, or an error if we
found a variable that was not in scope.

We implement the interpreter using a pair of mutually recursive
thunked functions. $\code{descend}$ descends into an input expression,
pushing frames onto the stack as it goes until it reaches a value
form. When a value form is reached, we $\code{ascend}$ with the
semantic value, interacting with the stack accordingly.
In the case that a variable is out of scope, we produce an error,
walking the stack popping off stack frames until we reach the
$\code{.end}$ and produce the error.
To run the abstract machine from a start state, the function
$\code{eval}$ provides an empty environment $\code{[]}$ and
an empty stack $\code{.end}$.

\section{Syntax and Semantics of Polymorphic CBPV}
\label{sec:zydeco-syntax-semantics}

In this section, we introduce our calculus
\cbpv{} an extension of CBPV
calculus with $F_\omega$-style impredicative polymorphism and
recursive computation types.

\subsection{Call-by-push-value Syntax}
\label{sec:zydeco-syntax}

A full overview of syntactic forms of
\cbpv{} is given in
Figure~\ref{fig:syntax-cbpv}. There are six basic syntactic
forms. Since we have $F_\omega$-style of polymorphism
\cite{girard72}, we have \emph{kinds}. In addition to function
kinds, we have two different kinds of types in CBPV. \emph{Value
types} (kind $\VTy$) classify inert data that can be passed around as
first-class values. \emph{Computation types} (kind $\CTy$) classify
imperative computations that interact with the stack, or dually, classify stack structures imperative computations interact with. Type environments $\Delta$ contain type variables of various
kinds. Value environments $\Gamma$ contain value variables with their
associated value types. Then the two types of terms are inert
first-class \emph{values} and imperative \emph{computations}.

\begin{figure}
  \begin{footnotesize}
    \[\begin{array}{rrcl}
  \name{Kind}        & K       & \Coloneqq & \VTy \mid \CTy \mid K_1 \to K_2 \\
  \\
  \name{Type Env}    & \Delta     & \Coloneqq & \CtxtEmpty \mid \CtxtExtend{\Delta}{X}{K} \\
  \name{Type}        & S, A, B & \Coloneqq & X \mid \SLam{X:K}{S} \mid S~S_0 \\
          &&\mid & \AThunk{B}
            \mid \AUnit
            \mid \AProd{A_1}{A_2}
            \mid \ACoProd{A_C}{c \in C}
            \mid \AExists{X:K}{A}
          \\
          &&\mid & \BReturn{A}
            \mid \BArrow{A}{B}
            \mid \BWith{B_d}{d \in D}
            \mid \BForall{X:K}{B}
            \mid \BNu{Y:K}{B}
          \\
  \\
  \name{Value Env}   & \Gamma     & \Coloneqq & \CtxtEmpty \mid \CtxtExtend{\Gamma}{x}{A} \\
  \name{Value}       & V       & \Coloneqq & x \mid \VThunk{M}
            \mid \VUnit
            \mid \VPair{V_1}{V_2}
            \mid \VInj{c}{V}
            \mid \VPack{S}{V}
            \mid \VHalt
            \\
  \name{Computation} & M       & \Coloneqq & \MForce{V}
            \mid \MLetPair{x_1}{x_2}{V}{M}
            \\
          &&\mid & \MMatch{V}{\Inj{c}(x_c)}{M_c}{c \in C}
          \\
          &&\mid & \MLetPack{X}{x}{V}{M}
          \\
          &&\mid & \MReturn{V} \mid \MBind{x}{M_0}{M}
            \mid \MTmLam{x:A}{M} \mid M~V
          \\
          &&\mid & \MComatch{~.d}{M_d}{d \in D} \mid M~.d
          \\
          &&\mid & \MTyLam{X:K}{M} \mid M~S
            \mid \MRoll{M} \mid \MUnroll{M} \mid \MFix{x}{M}
\end{array}\]
  \end{footnotesize}
\caption{Syntax of \cbpv{}}
\label{fig:syntax-cbpv}
\end{figure}

We present the kinding judgment in Figure~\ref{fig:kinding}. We have
standard rules for type variables and function kinds. Next we have the
value type constructors. These include unit, products, and sums,
which act as in a call-by-value language, as well as
existentially quantified packages $\AExists{X:K}{A}$ which can
quantify over types of any kind. 
A sum type $\ACoProd{A_c}{c \in C}$ is a labeled sum with constructor tags $c \in C$.
Finally, we have the type constructor
$\AThk$ of \emph{thunks} that takes a computation type to the type of
first class suspended computations of that type, i.e.,
\emph{closures}.

Next we have the computation type constructors, which classify how imperative computations can interact with the runtime
stack, and therefore, what the state of the stack may be.
The type $A \to B$ for $A$ a value type and $B$ a computation type is
the type of computations which are \emph{functions} that pop an $A$
value off of the stack and then proceed as a $B$
computation. Therefore the stack must consist of an $A$ value pushed
onto a $B$ stack.
Similarly, polymorphic functions $\forall X. B$ are functions that
take a type as input.
The type $\BWith{B_d}{d \in D}$ is a ``lazy'' product type of a collection of $B_d$ computation types indexed by a set of
destructor tags $d \in D$. The stacks for these consist of a destructor tag $.d$ followed by a $B_d$ stack.
As a special case, $\BTop$ is a nullary product and its stacks begin with a tag drawn from the empty set. In other words, there are no possible stacks, and so this type represents computations that are dead code.
The recursive computation types $\nu Y. B$ allow for recursive
specifications of computations, or dually, of recursive types of
stacks, which we will see allows us to define stacks that can be
arbitrarily large.
%
% Note: adding a comment here for recursive value types
It is also possible to have recursive value types, but we choose not to
include them because they are not needed for our examples.
Finally we have the type constructor $\BRet$ which takes a type of
values $A$ to the type of computations that can return values of type
$A$. Dually, the type of stacks is that of a continuation
for values of type $A$.

We have departed slightly from Levy's original CBPV syntax: inspired
by denotational models, he uses $U$ for $\AThk$ and $F$ for $\BRet$,
but we have chosen more operationally-motivated name. Additionally we
use the standard linear logic syntax for computation products.

\begin{figure}
\raggedright{
  \judgbox{$\ctxAssignKind{\Delta}{S}{K}$} ~~ type (constructor) $S$ has kind $K$ under type environment $\Delta$
}
\begin{footnotesize}
  \begin{mathpar}
  \judgment{
    X:K \in \Delta
  }{
    \Delta \vdash X:K
  }{TyTVar}
  \quad
  \judgment{
    \Delta, X:K_0 \vdash S:K
  }{
    \Delta \vdash \SLam{X:K_0}{S} : K_0 \to K
  }{TyTLam}
  \quad
  \judgment{
    \Delta \vdash S:K_0 \to K \\
    \Delta \vdash S_0:K_0
  }{
    \Delta \vdash S~S_0: K
  }{TyTApp}\\

  \AThk : \CTy \to \VTy

  \AUnit : \VTy

  (\times) : \VTy \to \VTy \to \VTy

  (+_C) : \multi{\VTy}_{c \in C} \to \VTy

  \judgment{
    \Delta, X:K_0 \vdash A:\VTy
  }{
    \Delta \vdash \AExists{X:K_0}{A} : \VTy
  }{TyExists}\\

  \BRet : \VTy \to \CTy

  (\to) : \VTy \to \CTy \to \CTy

  (\&_D) : \multi{\CTy}_{d \in D} \to \CTy

  \judgment{
    \Delta, X:K_0 \vdash B:\CTy
  }{
    \Delta \vdash \BForall{X:K_0}{B} : \CTy
  }{TyForall}

  \judgment{
    \Delta, Y:\CTy \vdash B:\CTy
  }{
    \Delta \vdash \BNu{Y:\CTy}{B} : \CTy
  }{TyNu}
\end{mathpar}
\end{footnotesize}
  \caption{Kinding Rules}
  \label{fig:kinding}
\end{figure}

Finally we give the typing rules for values and computations in
\cref{fig:value-typing,fig:computation-typing}.
First, we have the rules for values, which include variables, thunks,
which we denote with ``suspenders'' following the Frank language
\cite{lmm-2017}, and constructors for unit, product and sum types. The
elimination forms for these values, which are given by pattern
matching, are constructors on computations.

The \code{Ret} rules are similar to a monad, with \code{ret} as the
introduction rule and a bind for the elimination, but with an
arbitrary computation type allowed for the continuation.
We can force a \code{Thunk B} to get a \code{B} computation.
We add standard let and recursion rules, though note that the
recursive rule must put a \code{Thunk B} into the context, as all
variables are of value type.
Ordinary and polymorphic functions are given by $\lambda/\Lambda$ and
application rules.
Computation products have projections as their destructors and their
introduction rule is given by ``copattern match'' whose destructor
is on the stack \cite{copatterns}. Notationally, we use $\keyword{comatch}$ to
indicate the copattern match, and $M~.d$ to indicate the
application of the destructor $.d$ to the stack.
Additionally, $\MTop$ stands for the nullary copattern match.

\begin{figure}
\raggedright{
  \judgbox{$\ctxAssignType{\Delta;\Gamma}{V}{A}$}~~$V$ has type $A$ under type environment $\Delta$ and value environment $\Gamma$
}
\begin{footnotesize}
\begin{mathpar}
  \judgment{
    x:A \in \Gamma
  }{
    \ctxAssignType{\Delta;\Gamma}{x}{A}
  }{ValVar}

  \judgment{
    \ctxAssignType{\Delta;\Gamma}{M}{B}
  }{
    \ctxAssignType{\Delta;\Gamma}{\VThunk{M}}{\AThunk{B}}
  }{ValThunk}

  \VUnit : \AUnit

  \judgment{
    \ctxAssignType{\Delta;\Gamma}{V_1}{A_1} \\
    \ctxAssignType{\Delta;\Gamma}{V_2}{A_2}
  }{
    \ctxAssignType{\Delta;\Gamma}{\VPair{V_1}{V_2}}{\AProd{A_1}{A_2}}
  }{ValPair}

  \judgment{
    \ctxAssignType{\Delta;\Gamma}{V}{A_c}
  }{
    \ctxAssignType{\Delta;\Gamma}{\VInj{c}{V}}{\ACoProd{A_c}{c \in C}}
  }{ValInj}

  \judgment{
    \ctxAssignType{\Delta;\Gamma}{V}{\subst{A}{S}{X}}
  }{
    \ctxAssignType{\Delta;\Gamma}{\VPack{S}{V}}{\AExists{(X:K)}{A}}
  }{ValPack}
\end{mathpar}
\end{footnotesize}
\caption{Value Typing for \cbpv{}}
\label{fig:value-typing}
\end{figure}

\begin{figure}
\raggedright{
  \judgbox{$\ctxAssignType{\Delta;\Gamma}{M}{B}$}~~$M$ has type $B$ under type environment $\Delta$ and value environment $\Gamma$
}
\begin{footnotesize}
  \begin{mathpar}
  \judgment{
    \ctxAssignType{\Delta;\Gamma}{V}{\AThunk{B}}
  }{
    \ctxAssignType{\Delta;\Gamma}{\MForce{V}}{B}
  }{ComForce}

  \judgment{
    \ctxAssignType{\Delta;\Gamma}{V}{A} \\
    \ctxAssignType{\Delta;\Gamma, x:A}{M}{B}
  }{
    \ctxAssignType{\Delta;\Gamma}{\MLet{x}{V}{M}}{B}
  }{ComLet}

  \judgment{
    \ctxAssignType{\Delta;\Gamma}{V}{\AProd{A_0}{A_1}} \\
    \ctxAssignType{\Delta;\Gamma, x_0:A_0, x_1:A_1}{M}{B}
  }{
    \ctxAssignType{\Delta;\Gamma}{\MLetPair{x_0}{x_1}{V}{M}}{B}
  }{ComPrj}

  \judgment{
    \ctxAssignType{\Delta;\Gamma}{V}{\ACoProd{A_c}{c \in C}} \\
    \forall c \in C, \ \ \ctxAssignType{\Delta;\Gamma, x_c:A_c}{M_c}{B}
  }{
    \ctxAssignType{\Delta;\Gamma}{\MMatch{V}{\Inj{c}(x_c)}{M_c}{c \in C}}{B}
  }{ComMatch}

  \judgment{
    \ctxAssignType{\Delta;\Gamma}{V}{\AExists{(Y:K)}{A}} \\
    \ctxAssignType{\Delta, X:K;\Gamma, x:\subst{A}{X}{Y}}{M}{B}
  }{
    \ctxAssignType{\Delta;\Gamma}{\MLetPack{X}{x}{V}{M}}{B}
  }{ComUnpack}

  \judgment{
    \ctxAssignType{\Delta;\Gamma}{V}{A}
  }{
    \ctxAssignType{\Delta;\Gamma}{\MReturn{V}}{\BReturn{A}}
  }{ComReturn}

  \judgment{
    \ctxAssignType{\Delta;\Gamma}{M_0}{\BReturn{A}} \\
    \ctxAssignType{\Delta;\Gamma, x:A}{M}{B}
  }{
    \ctxAssignType{\Delta;\Gamma}{\MBind{x}{M_0}{M}}{B}
  }{ComBind}

  \judgment{
    \ctxAssignType{\Delta;\Gamma, x:A}{M}{B}
  }{
    \ctxAssignType{\Delta;\Gamma}{\MTmLam{(x:A)}{M}}{\BArrow{A}{B}}
  }{ComLam}

  \judgment{
    \ctxAssignType{\Delta;\Gamma}{M}{\BArrow{A}{B}} \\
    \ctxAssignType{\Delta;\Gamma}{V}{A}
  }{
    \ctxAssignType{\Delta;\Gamma}{M~V}{B}
  }{ComApp}

  \judgment{
    \forall d \in D, \ \ \ctxAssignType{\Delta;\Gamma}{M_d}{B_d}
  }{
    \ctxAssignType{\Delta;\Gamma}{(\MComatch{.d}{M_d}{d \in D})}{\BWith{B_d}{d \in D}}
  }{ComCoMatch}

  \judgment{
    \ctxAssignType{\Delta;\Gamma}{M}{\BWith{B_d}{d \in D}}
  }{
    \ctxAssignType{\Delta;\Gamma}{M~.d}{B_d}
  }{ComDtor}

  \judgment{
    \ctxAssignType{\Delta, Y:K;\Gamma}{M}{\subst{B}{Y}{X}}
  }{
    \ctxAssignType{\Delta;\Gamma}{\MTyLam{(Y:K)}{M}}{\BForall{(X:K)}{B}}
  }{ComTyLam}

  \judgment{
    \ctxAssignType{\Delta;\Gamma}{M}{\BForall{(X:K)}{B}}
  }{
    \ctxAssignType{\Delta;\Gamma}{M~S}{\subst{B}{S}{X}}
  }{ComTyApp}

  \judgment{
    \ctxAssignType{\Delta;\Gamma}{M}{\subst{B}{\BNu{(Y:K)}{B}}{Y}}
  }{
    \ctxAssignType{\Delta;\Gamma}{\MRoll{M}}{\BNu{(Y:K)}{B}}
  }{ComRoll}

  \judgment{
    \ctxAssignType{\Delta;\Gamma}{M}{\BNu{(Y:K)}{B}}
  }{
    \ctxAssignType{\Delta;\Gamma}{\MUnroll{M}}{\subst{B}{\BNu{(Y:K)}{B}}{Y}}
  }{ComUnroll}

  \judgment{
    \ctxAssignType{\Delta;\Gamma, x:\AThunk{B}}{M}{B}
  }{
    \ctxAssignType{\Delta;\Gamma}{\MFix{x}{M}}{B}
  }{ComFix}
\end{mathpar}
\end{footnotesize}
\caption{Computation Typing for \cbpv{}}
\label{fig:computation-typing}
\end{figure}

\subsection{Operational Semantics with a Stack Machine}
\label{sec:zydeco-opsem}

% Stack Machine
\newcommand{\machine}[2]{\ensuremath{\langle #1 \mid #2 \rangle}}
\newcommand{\machineStep}[2]{\ensuremath{#1 \longrightarrow #2}}
\newcommand{\assignSubst}[2]{\ensuremath{#1 / #2}}

We present the operational semantics of \cbpv{} in \cref{fig:opsem},
using an abstract stack machine.
First we inductively define stacks to be either an empty stack
$\SAmbient$ or a stack frame for one of the computation types. The
$\BRet$ frame is a continuation, application and type application
frames have an argument pushed on the stack. The computation product
has a destructor tag pushed onto the stack and the recursive
computation has an unroll on the stack. We note that the type
application and unroll are not strictly necessary at runtime but make
it easier to keep track of the typing of the stack. An initial state
is a returning computation running against the empty stack
$\machine{M}{\SAmbient}$ and the computation terminates if it reaches
a state $\machine{\MReturn{V}}{\SAmbient}$.

The rules that appear in \cref{fig:opsem} are standard for a CBPV
calculus with a stack machine, similar to Levy's original stack
machine semantics \cite{levy-stack-2005}.
Each of the computation elimination rules operates by pushing a frame
onto the stack, and each computation introduction rule operates by
popping the stack frame off and proceeding in a manner dependent on
the contents of the frame.
Forcing a thunk executes the suspended computation within.
Let, fixed points and pattern matching are straightforward.

\begin{figure}
\begin{footnotesize}
\[\begin{array}{rrcl}
  \name{Stack}       & \Stk       & \Coloneqq & \SAmbient \mid \SKont{x}{M}{\Stk} \mid \SApp{V}{\Stk} \mid \SApp{V}{\Stk} \mid \SDtor{i}{\Stk} \mid \SUnroll{\Stk}
\end{array}\]  
\end{footnotesize}

\raggedright{
  \judgbox{\machineStep{\machine{M}{\Stk}}{\machine{M'}{\Stk'}}}
  ~~$M$ with stack $\Stk$ steps to $M'$ with stack $\Stk'$
}
\begin{footnotesize}
\begin{mathpar}
  \judgment{
  }{
    \machineStep{
      \machine{\MReturn{V}}{\SKont{x}{M}{\Stk}}
    }{
      \machine{\subst{M}{V}{x}}{\Stk}
    }
  }{OpRet}

  \judgment{
  }{
    \machineStep{
      \machine{\MBind{x}{M_1}{M_2}}{\Stk}
    }{
      \machine{M_1}{\SKont{x}{M_2}{\Stk}}
    }
  }{OpBind}

  \judgment{
  }{
    \machineStep{
      \machine{\MForce{\VThunk{M}}}{\Stk}
    }{
      \machine{M}{\Stk}
    }
  }{OpForceThunk}

  \judgment{
  }{
    \machineStep{
      \machine{\MLet{x}{V}{M}}{\Stk}
    }{
      \machine{\subst{M}{V}{x}}{\Stk}
    }
  }{OpLet}

  \judgment{
  }{
    \machineStep{
      \machine{\MLetPair{x_0}{x_1}{\VPair{V_0}{V_1}}{M}}{\Stk}
    }{
      \machine{\subst{\subst{M}{V_0}{x_0}}{V_1}{x_1}}{\Stk}
    }
  }{OpPrj}

  \judgment{
  }{
    \machineStep{
      \machine{\MMatch{\VInj{c}{V}}{\Inj{c}(x_c)}{M_c}{c \in C}}{\Stk}
    }{
      \machine{\subst{M_c}{V}{x_c}}{\Stk}
    }
  }{OpMatch}

  \judgment{
  }{
    \machineStep{
      \machine{\MLetPack{X}{x}{\VPack{S}{V}}{M}}{\Stk}
    }{
      \machine{\subst{M}{V}{x}}{\Stk}
    }
  }{OpUnpack}

  \judgment{
  }{
    \machineStep{
      \machine{M~V}{\Stk}
    }{
      \machine{M}{\SApp{V}{\Stk}}
    }
  }{OpApp}

  \judgment{
  }{
    \machineStep{
      \machine{\MTmLam{x}{M}}{\SApp{V}{\Stk}}
    }{
      \machine{\subst{M}{V}{x}}{\Stk}
    }
  }{OpLam}\quad
  \judgment{
  }{
    \machineStep{
      \machine{M~.d}{\Stk}
    }{
      \machine{M}{\SDtor{d}{\Stk}}
    }
  }{OpDtor}

  \judgment{
  }{
    \machineStep{
      \machine{(\MComatch{.d}{M_d}{d \in D})}{\SDtor{d}{\Stk}}
    }{
      \machine{M_d}{\Stk}
    }
  }{OpComatch}

  \judgment{
  }{
    \machineStep{
      \machine{M~S}{\Stk}
    }{
      \machine{M}{\STyApp{S}{\Stk}}
    }
  }{OpTyApp}\quad
  \judgment{
  }{
    \machineStep{
      \machine{\MTyLam{X}{M}}{\STyApp{S}{\Stk}}
    }{
      \machine{\subst{M}{S}{X}}{\Stk}
    }
  }{OpTyLam}

  \judgment{
  }{
    \machineStep{
      \machine{\MRoll{M}}{\SUnroll\Stk}
    }{
      \machine{M}{\Stk}
    }
  }{OpRoll}

  \judgment{
  }{
    \machineStep{
      \machine{\MUnroll{M}}{\Stk}
    }{
      \machine{M}{\SUnroll\Stk}
    }
  }{OpUnroll}

  \judgment{
  }{
    \machineStep{
      \machine{\MFix{x}{M}}{\Stk}
    }{
      \machine{M[\VThunk{\MFix{x}{M}} / x]}{\Stk}
    }
  }{OpFix}
\end{mathpar}
\end{footnotesize}
\caption{Operational Semantics of \cbpv{} via a Stack Machine}
\label{fig:opsem}
\end{figure}

\section{Relative Monads in CBPV}
\label{sec:relative-monads-in-zydeco}

Consider our example implementation of a polynomial function $x^2+x+c$
operating on fixed-precision integers from \cref{sec:cbpv-intro}.  How
would we change this code if we wanted to produce an \emph{error} if
any of the three arithmetic operations were to result in an overflow?
Haskell programmers know well that this can be achieved easily by
using an error monad such as $\code{Either}~\code{e}~\code{a} = \ACoProdLit{\code{Left}:\code{e},~\code{Right}:\code{a}}$. That is, we
would change to versions of our $\code{times}$ and $\code{add}$
operations that return not an $\AInt$ but an
$\code{Either}~\AString~\AInt$, and we would chain the calls together
using a monadic bind, maintaining the style of the original code, but
implicitly short-circuiting and producing an error if any of the three
calls produces an error.  How can we adapt this pattern to a CBPV
setting? The output of $\code{add}$ and $\code{times}$ is not just
$\AInt$ but $\BReturn{\AInt}$. Operationally, the $\BReturn{\AInt}$
type indicates that the stack is a continuation that we can return
$\AInt$s to. To produce erroring computations, we would want to
replace this with a stack that can support not just returning $\AInt$
values, but also producing errors with an error message given as a
$\AString$. So we want a type constructor $T : \VTy \to \CTy$ that
takes a type of values $A$ to the type of computations that either error
or produce $A$ values.

The interface that this type needs to satisfy is the following
modification of Haskell's monad typeclass:
\begin{footnotesize}
  \begin{align*}
    \BMo~(\code{T} : \VTy\to\CTy) ~\defeq~&\{\,\,\code{.return}~: ~\BForall{(\code{A}: \VTy)}{\code{A} \to \code{T}~\code{A}} \\
          &\,\,\,\,\,\code{.bind}~: ~\BForall{(\code{A}~\code{A'}: \VTy)} {\AThunk{(\code{T}~\code{A})} \to \AThunk{(\code{A} \to \code{T}~\code{A'})} \to \code{T}~\code{A'}} \,\,\}
\end{align*}
\end{footnotesize}
Here, our notion of relative monad is parameterized by a type
constructor $\code{T}$ that takes a type of values $\code{A}$ to the type of $\code{T}~\code{A}$
computations that perform $\code{T}$-effects and return $\code{A}$-values.
The \code{.return} and \code{.bind} operations are our adaptation of
the usual constructors for a monad. The \code{.return} operation specifies how to
construct a pure $\code{T}~\code{A}$ computation from an $\code{A}$ value. The \code{.bind} operation
specifies how to sequentially compose effectful computations: given a
$\AThunk(\code{T}~\code{A})$ computation and a continuation $\AThunk(\code{A} \to \code{T}~\code{A'})$
that takes in $\code{A}$ values and produces a $\code{T}~\code{A'}$ computation, we can
construct a $\code{T}~\code{A'}$ computation that performs the computations in
sequence, feeding the results of the first as inputs into the second.

$\BMo~\code{T}$ is not a \emph{monad} in the usual sense. A monad in category
theory is an endofunctor, meaning its corresponding type constructor
$\code{T}$ should return the same \emph{kind} of types that it takes
in. Instead $\BMo~\code{T}$ fits into the more general category-theoretic
notion of a \emph{relative} monad, which we explain in
\cref{sec:fund-thm}.

We can also consider relative monads in terms of our dual view of
computation types as encoding a type of stacks. In this view,
$\BReturn{\code{A}}$ is the built-in notion of a continuation, so a relative
monad $\code{T}~\code{A}$ should then be thought of as a \emph{user-defined
notion of continuation}.
%
% Note: change the explanation here according to the review comments
% Note: but not much; not sure if it's clear enough to the reader
In this view, \code{.return} specifies how to return an $\code{A}$ value to a $\code{T}~\code{A}$ continuation. 
Second, \code{.bind} is viewed in reverse as taking in a $\code{T}~\code{A'}$
continuation and composing it with a function $\AThunk(\code{A} \to \code{T}~\code{A'})$
that takes in $\code{A}$ values and interacts with a $\code{T}~\code{A'}$ computation to
produce a continuation for $\code{A}$ values.
%
%% The \code{.return} and \code{.bind} operations are then the two
%% fundamental operations that such a continuation needs to support.
%% First, \code{.return} says how we can return a value to the current
%% continuation. Second, \code{.bind} gives us a way to \emph{extend} a
%% $T A'$ continuation to a $T A$ continuation by providing a function
%% from $A$ that uses a $T A'$ continuation.

Returning to our example, if our $\code{add}$ and $\code{times}$ were
re-written to output $\code{T}~\AInt$ rather than $\BReturn{\AInt}$, we could
re-write our polynomial computation by replacing the built-in
\code{do} notation with calls to $\MForce{\code{m}}~\code{.bind}$ given a
monad implementation $\code{m} : \AThunk{(\BMonad{T})}$:
\begin{footnotesize}
  \begin{align*}
  \MTmLam{\code{x}}&!~\code{m}~\AInt~\AInt~\code{.bind}~\{~!~\code{times}~\code{x}~\code{x}~\}~\{~\MTmLam{\code{s}}\\
    &!~\code{m}~\AInt~\AInt~\code{.bind}~\{~!~\code{add}~\code{x}~10~\}~\{~\MTmLam{\code{y}} \\
  &!~\code{add}~\code{s}~\code{y}~
  \}\}
\end{align*}
\end{footnotesize}
We provide an analog to Haskell's \keyword{do} notation in
\cref{sec:mo-blocks-and-alg-trans}.

\subsection{Examples}

The first example of a relative monad is the simplest one: $\BRet$ itself.
\begin{footnotesize}
  \begin{align*}
    &\code{mret} : \AThunk{(\BMonad{\BRet})}~\defeq~\{~\keyword{comatch} \\
      &\qquad\qquad\begin{aligned}
        &|~\code{.return}~\Rightarrow~\MTyLam{\code{A}}{\MTmLam{\code{a}}{\keyword{ret}~\code{a}}} \\
        &|~\code{.bind}~\Rightarrow~\MTyLam{\code{A}~\code{A'}}{\MTmLam{\code{m}~\code{f}}{
          \keyword{do}~\code{a}~\leftarrow~!~\code{m};
          ~!~\code{f}~\code{a}~}}\}
      \end{aligned}
\end{align*}
\end{footnotesize}
The \code{return} and \code{bind} for this monad are essentially built
into CBPV.  In fact, the $\BRet$ monad plays the same role in CBPV
that the \code{Identity} monad does in Haskell, representing ``no
effects'' or at least no \emph{more} effects than the ambient notion
of effect allowed in the language.

Next in Figure~\ref{fig:three-exception-monads}, we provide three different ways we could implement a monad that
would model the erroring computations from our motivating example.
First, \code{Exn} is a direct translation of Haskell's \code{Either}
type: an erroring computation is one that returns either an error or a
successful value.
Then we can implement \code{return} which specifies how a successful
computation can be embedded in a possibly erroring one. Finally we
implement a monadic \code{bind} that tells us how to extend an
erroring computation with a continuation for its succesful values: if
an error value is produced, we return it immediately, discarding the
continuation. Otherwise if a successful value is produced we execute
the continuation on it.
This formulation of exceptional computations has well known
performance issues: every \code{bind} pattern matches on the sum,
causing branching even though the exception ``handler'' doesn't
change.

The downsides of the first design are addressed by the second
implementation: the Church-encoded exception type \code{ExnK}. This is
the well-known technique of ``double-barreled'' continuation-passing
\cite{double-barrel}. This type
uses computation types non-trivially: the stack for an erroring
computation contains a pair of a continuation, one for errors and one
for successful return values.

The stack for this computation is
one that contains two continuations, and a tail of the stack which has
an abstract type $\code{R}$. Parametricity of the language ensures that this
result type $\code{R}$ cannot leak any information
\cite{parametricity-eec}.
We can easily tell from the implementation of the monad interface that,
unlike the previous solution, no intermediate data constructor is
produced, eliminating unnecessary conditional branching as the control
flow is carried out by invoking the correct continuation directly.
As we are in a continuation-passing style now, the stack grows only by
the continuations themselves growing bigger and bigger along the way,
eventually being run when a return or raise is executed.

A third method for implementing raising of exceptions is via explicit
stack-walking to find the nearest enclosing exception handler, which
we model with $\code{ExnDe}$.
We can encode this structure as well in CBPV by using a coinductive
codata type, which we can view as ``defunctionalizing'' the
continuations into explicit stack frames
\cite{reynolds-definitional-1972,wand-1980}.
In this version, the exception type is given as a recursive product
computation type, meaning the stack ends up being of arbitrarily large
size.
Instead of extending the current continuations on the stack, we
implement success continuations and handlers by pushing \code{.kont}
and \code{.try} continuations onto the stack directly before making a
call.
Then exception raising code walks the stack to find the nearest
enclosing handler. If desired this method could be adjusted to include
an error value that the exception handles, so that exception raising
code can skip over handlers that don't explicitly catch that type of
exception.

\begin{figure}
  \begin{minipage}{0.35\textwidth}\begin{footnotesize}
    \begin{align*}
    \code{Exn}~&\defeq~\SLam{(\code{E}~\code{A}: \VTy)}{
        \BReturn{(\ACoProdLit{\code{Err}:\code{E}, \code{Ok}:\code{A}})}
    }\\
    \code{mexn}~&\defeq~\{~\MTyLam{(\code{E}: \VTy)}{\keyword{comatch} \\
      &\qquad\quad\begin{aligned}
        &|~\code{.return}~\Rightarrow~\MTyLam{\code{A}}~\MTmLam{\code{a}}~\keyword{ret}~(\code{Ok}(\code{a})) \\
        &|~\code{.bind}~\Rightarrow~\MTyLam{\code{A}~\code{A'}}\MTmLam{\code{m}~\code{f}} \\
        &\qquad\keyword{do}~\code{a?}~\leftarrow~!~\code{m}; \\
        &\qquad\keyword{match}~\code{a?} \\
        &\qquad|\,\,\,\code{Err}(\code{e})~\Rightarrow~\keyword{ret}~(\code{Err}(\code{e})) \\
        &\qquad|\,\,\,\code{Ok}(\code{a})~\Rightarrow~!~\code{f}~\code{a}~\} \\
      \end{aligned}
    }
    \end{align*}
  \end{footnotesize}\end{minipage}
  \begin{minipage}{0.45\textwidth}\begin{footnotesize}
    \begin{align*}
    \code{ExnK}~&\defeq~\SLam{(\code{E}~\code{A}: \VTy)}{
      \BForall{\code{R}:\CTy} \\ &\BArrow{\AThk{(\BArrow{\code{E}}{\code{R}})}}{\BArrow{\AThk{(\BArrow{\code{A}}{\code{R}})}}{\code{R}}}}\\
    \code{mexnk}~&\defeq~\{~\MTyLam{(\code{E}: \VTy)}{\keyword{comatch} \\
        &\qquad\quad\begin{aligned}
          &|~\code{.return}~\Rightarrow \MTyLam{\code{A}}{\MTmLam{\code{a}}{ \\
          &\qquad\MTyLam{\code{R}}{\MTmLam{\code{ke}~\code{ka}}{!~\code{ka}~\code{a}}}}} \\
          &|~\code{.bind}~\Rightarrow \MTyLam{\code{A}~\code{A'}}{\MTmLam{\code{m}~\code{f}}{ \\
            &\qquad\MTyLam{\code{R}}{\MTmLam{\code{ke}~\code{ka}}{!~\code{m}~\code{R}~\code{ke} \\
            &\qquad\quad\VThunk{~\MTmLam{\code{a}}{!~\code{f}~\code{a}~\code{R}~\code{ke}~\code{ka}}~}
            }}
          }}~\}
        \end{aligned}
      }
    \end{align*}
  \end{footnotesize}\end{minipage}
  \\
  \begin{footnotesize}
    \begin{align*}
    \code{ExnDe}~&\defeq~\nu \code{ExnDe}.\SLam{(\code{E}~\code{A}: \VTy)}{ \\
      &\qquad\quad\begin{aligned}
        &  \&\{\,\,\code{.try}: \BForall{(\code{E'}: \VTy)}{\AThunk{(\code{E} \to \code{ExnDe}~\code{E'}~\code{A})} \to \code{ExnDe}~\code{E'}~\code{A}} \\
        & \quad\,\,\,\code{.kont}: \BForall{(\code{A'}: \VTy)}{\AThunk{(\code{A} \to \code{ExnDe}~\code{E}~\code{A'})} \to \code{ExnDe}~\code{E}~\code{A'}} \\
        & \quad\,\,\,\code{.done}: \BReturn{(\ACoProdLit{\code{Err}:\code{E},~\code{Ok}:\code{A}})}\,\,\}
      \end{aligned}
    }\\
    \code{mexnde}~&\defeq~\{~\MFix{\code{mexnde}}{\MTyLam{(\code{E}: \VTy)}{\keyword{comatch} \\
        &\qquad\quad\begin{aligned}
          &|~\code{.return}~\Rightarrow \MTyLam{\code{A}}\MTmLam{\code{a}}~\keyword{comatch} \\
          &\qquad|~\code{.try}~\Rightarrow \MTyLam{\code{E'}}{\MTmLam{\code{k}}{!~\code{mexnde}~\code{E'}~\code{.return}~\code{A}~\code{a}}} \\
          &\qquad|~\code{.kont}~\Rightarrow \MTyLam{\code{A'}}{\MTmLam{\code{k}}{!~\code{k}~\code{a}}} \\
          &\qquad|~\code{.done}~\Rightarrow \keyword{ret}~(\code{Ok}(\code{a})) \\
          &|~\code{.bind}~\Rightarrow \MTyLam{\code{A}~\code{A'}}{\MTmLam{\code{m}~\code{f}}{!~\code{m}~\code{.kont}~\code{A'}~\code{f}}}~\}
        \end{aligned}
    }}
  \end{align*}
  \end{footnotesize}
  \caption{Three Exception Monads}
  \label{fig:three-exception-monads}
\end{figure}

%% Comment: perhaps you could just say here that you can do a bunch of
%% other useful things, and here's their types, we omit the
%% implementation which simply follows the types.
%
Next, we show that several other typical monads can be adapted to
CBPV.  For brevity, we give only the types of these monads. The full
details of their implementations are given in
\ifextended
\cref{sec:appendix:monad-implementations}.
\else
the extended version of the paper \cite{extended-version}.
\fi
%
% For remaining examples, for brevity we give only the type signatures
% of the monads that one can implement in CBPV and omit the
% implementations of the monads due to their similarity to the examples
% of exception monad shown, but they are included in the supplementary
% material.

We give the types of the two continuation monads \code{Kont} and
\code{PolyKont} and state monads \code{State} and \code{StateK} in
\cref{fig:mo_state_kont}. The first continuation monad \code{Kont} is
the one with fixed answer type \code{B}. The intuition is that a stack
for \code{Kont} is a \code{B} stack with a continuation on top. Since
\code{B} is fixed, the computation can interact with the stack before
invoking the continuation (if ever). For instance, in our Zydeco
implementation, we include a built-in computation type \code{OS} that
is used for computations that can make system calls by using a
continuation-passing style, and $\code{Kont}\,\code{OS}$ then serves
as an analog of Haskell's \code{IO} monad.

The second continuation monad \code{PolyKont} looks similar at first
to \code{Kont} but the result parameter \code{R} is this time
\emph{universally quantified}, as in the double-barreled continuation
example. Since the \code{R} is universally quantified, by
parametricity the computation cannot interact with the \code{R} stack
\emph{except} by calling the continuation.
For this reason, using parametric reasoning it can be proven that
\code{PolyKont} is equivalent to $\BRet$, i.e., it is a
Church-encoding of the $\BRet$ type
\cite{reynolds-1983,parametricity-eec}.
For this reason it is technically possible not to include $\BRet$ as a
primitive type, replacing its uses with \code{PolyKont}.

Next, we provide two versions of the state monad. One is a
``direct-style'' version, which takes in an input state and returns a
new state alongside the returned value. The second is a Church-encoded
version, where instead a continuation is passed in which expects the
state to be passed on the stack. 

\begin{figure}
  \centering
  \begin{footnotesize}
    \begin{align*}
    \code{Kont}~&\defeq~\SLam{(\code{B}: \CTy)~(\code{A}: \VTy)}
      \AThunk{(\code{A} \to \code{B})} \to \code{B}
    \\
    \code{PolyKont}~&\defeq~\SLam{(\code{A}: \VTy)}
      \BForall{(\code{R}: \CTy)}{\AThunk{(\code{A} \to \code{R})} \to \code{R}}
    \\
    \code{State}~&\defeq~\SLam{(\code{S}: \VTy)~(\code{A}: \VTy)}
      \code{S} \to \BRet{(\code{A} \times \code{S})}
    \\
    \code{StateK}~&\defeq~\SLam{(\code{S}: \VTy)~(\code{A}: \VTy)}
      \BForall{(\code{R}: \CTy)}{\AThunk{(\code{A} \to \code{S} \to \code{R})} \to \code{S} \to \code{R}}
  \end{align*}
  \end{footnotesize}
  \caption{Continuation and State Monads}
  \label{fig:mo_state_kont}
\end{figure}

Finally, we turn to a general class of monads, \emph{free monads}
generated by a collection of operations, in the sense of algebraic
effects \cite{plotkin-power-2001}.
Free monads give a way to take an arbitrary collection of desired
syntactic ``operations'' and provide a monad \code{T} that supports
those operations. Here each operation is specified by two value types:
a domain \code{A} and codomain \code{A'}. Then the free monad \code{T}
should be the minimal monad supporting for each operation a function
$\AThk{(\BArrow{\code{A}}{\code{T}~\code{A'}})}$.
For instance, a printing operation would be an operation
$\AThk{(\BArrow{\AString}{\code{T}~\AUnit})}$ or a random choice operation
would be $\AThk{(\BArrow{\AUnit}{\code{T}~\ABool})}$.
Typically free monads are constructed as a type of \emph{trees}
inductively generated by return as a leaf of the monad, and a node
constructor for each operation.
This construction could be carried out to construct a monad using
recursive value types, and then compose this with $\BRet$ to get a
relative monad. On the other hand, we can construct this directly as a
computation type using a similar Church-encoding trick to what we have
seen thus far.
For instance, here is the type constructor for a free monad supporting
print and random choice operations, which can obviously be generalized
to any number of operations:
\begin{footnotesize}
  \begin{align*}
  \code{KPrint}~&\defeq~\SLam{\code{R}: \CTy}{\AThunk{(\AString \to \code{Kont}~\code{R}~\AUnit)}} \\
  \code{KFlip}~&\defeq~\SLam{\code{R}: \CTy}{\AThunk{(\AUnit \to \code{Kont}~\code{R}~\ABool)}} \\
  \code{PrintFlip}~&\defeq~\SLam{\code{A}:\VTy}{\BForall{\code{R}: \CTy}{\code{KPrint}~\code{R} \to \code{KFlip}~\code{R} \to \code{Kont}~\code{R}~\code{A}}}
\end{align*}
\end{footnotesize}
An intuition for this construction is that a computation in the free
monad is one that takes in any ``interpretation'' of the operations as
acting on some result type \code{R} and an \code{A}-continuation for
\code{R} and returns an \code{R}. Again polymorphism ensures that all
such a computation can do is perform some sequence of operations and
return.
We provide the full definition of the monad structure for
\code{PrintFlip} as well as demonstrate that it supports \code{Print}
and \code{Flip} operations in
\ifextended
\cref{sec:appendix:monad-implementations}
\else
the extended version of the paper \cite{extended-version}.
\fi
.

\subsection{Laws}
Just as with monads, there are relative monad \emph{laws}.
\begin{definition}
  A relative monad implementation $\code{m} : \AThk{(\BMo~\code{T})}$ is \emph{lawful} if it satisfies the following equational principles:
  \begin{enumerate}
  \item \textbf{left unit law:} for any $\code{a}: \code{A}$ and $\code{f}: \AThunk{(\code{A} \to \code{T}~\code{A'})}$, \begin{footnotesize}\[
    \MForce{m~\code{.bind}}~\code{A}~\code{A'}~\VThunk{~\MForce{\code{m}~\code{.return}}~\code{A}~\code{a}~}~\code{f} \equiv \MForce{\code{f}}~\code{a}
  \]\end{footnotesize}
  \item \textbf{right unit law:} for any $\code{t}: \AThunk{(\code{T}~\code{A})}$, \begin{footnotesize}\[
    \MForce{\code{m}~\code{.bind}}~\code{A}~\code{A}~\code{t}~\VThunk{~\MForce{\code{m}~\code{.return}}~\code{A}~} \equiv \MForce{\code{t}}
  \]\end{footnotesize}
  \item \textbf{associativity law:} for any $\code{f}: \AThunk{(\code{A} \to \code{T}~\code{A'})}$, $g: \AThunk{(\code{A'} \to T~\code{A''})}$, and $t: \AThunk{(T~\code{A})}$, \begin{footnotesize}\[
    \MForce{\code{m}~\code{.bind}}~\code{A'}~\code{A''}~\VThunk{~\MForce{\code{m}~\code{.bind}}~\code{A}~\code{A'}~\code{t}~\code{f}~}~\code{g} \equiv \MForce{\code{m}~\code{.bind}}~\code{A}~\code{A'}~\code{t}~\VThunk{~\MTmLam{\code{x}}{\MForce{\code{m}~\code{.bind}}~\code{A'}~\code{A''}~\VThunk{~\MForce{\code{f}}~\code{x}~}~\code{g}}~}
  \]\end{footnotesize}
  \item \textbf{linearity law:} for any $\code{tt}: \AThunk{(\BReturn{(\AThunk{(\code{T}~\code{A})})})}$, \begin{footnotesize}\[
    \MBind{\code{t}}{\MForce{\code{tt}}}{\MForce{\code{m}~\code{.bind}}~\code{A}~\code{A'}~\code{t}} \equiv \MForce{\code{m}~\code{.bind}}~\code{A}~\code{A'}~\VThunk{~\MBind{\code{t}}{\MForce{\code{tt}}}{\MForce{\code{t}}}~}
  \]\end{footnotesize}
  \end{enumerate}  
\end{definition}
The first three of these are analogous to the ordinary monad laws. The
left unit law says that returning to a continuation constructed by
\code{bind} is the same as calling the continuation directly. The
right unit law says that constructing a continuation that always
returns is equivalent to using the ambient continuation. The
associativity law says that composing multiple continuations with bind
is associative.
The fourth, the linearity law, is a new law that is necessary in a
system like CBPV which may have ambient effects, as opposed to
ordinary monads which are defined in a ``pure'' language.
Intuitively, bind being linear states that it uses its first input
\emph{strictly} and never again. For instance, if bind simply pushes
something onto the stack and executes the thunk, then it will be
linear.
The precise formulation used here says that for \code{tm} a ``double
thunk'', it doesn't matter if we evaluate \code{tm} first, and bind
the resulting thunk or if we bind a thunk that will evaluate \code{tm}
each time it is called. This formulation is taken from prior work on
CBPV \cite{levy-2017,munch-linearity}.
One motivation for why the linearity rule is necessary is that we want
a relative monad to ``overload'' not just the \code{do} notation that
the primitive $\BRet$ type provides, we also want all of the
\emph{equational reasoning principles} that the $\BRet$ provides to
hold for relative monads. And the built-in $\code{do}$ for $\BRet$
simply pushes a continuation onto the stack and executes the input, so
it satisfies the linearity requirement.

Most of our example monads can be easily verified to satisfy the monad
laws, using CBPV $\beta\eta$ equality, and in the case of the
Church-encoded types, parametricity. However, the defunctionalized
exception monad implementation does not!
The reason is that by defunctionalizing the stack of continuations, we
expose low-level details and so ``non-standard'' computations can
inspect the stack and observe, for instance, the number of frames
pushed onto the stack.
As a result, while the left unit and linearity laws hold, the right
unit and associativity laws fail. We provide an explicit
counter-example in the \ifextended{appendix}\else{extended version}\fi which is implemented via stack-walking
and counts the number of stack frames that it encounters\ifextended{}\else{ \cite{extended-version}}\fi.

%% , with a counter-example provided in \cref{fig:exn_kont_counter}.
%% This counter example counter-example works similarly to \code{abort}, looping through and counting the number of stack frames using an accumulator.
%% To show observational difference, we can pass to \code{bench} two functions that
%% should be equal under the left unit law or the associativity law, and look at the program output.
%% `! count-kont`, as its name indicates,
%% performs a stack walk and distinguishes \code{impl}s by their number of \code{.kont} frames.

The root cause of this lawlessness is that problematic programs involve unrestricted manipulations on the stack frame,
yet a canonical instance of the exception monad, when executed, should only compute its result---either
an exception or a valid value---and react to one of the destructors correspondingly.
We henceforth propose a canonicity condition for these instances as shown below,
where \code{M} is an arbitrary computation of type $\code{Exn}~\code{E}~\code{A}$.
\begin{footnotesize}\begin{align*}
  \code{M} \equiv~&\MBind{\code{x}}{\code{M}}{} \\
           &\keyword{match}~\code{x} \\
           &|\,\,\,\code{Err}(\code{e}) \Rightarrow \MForce{\code{fail}~\code{e}} \\
           &|\,\,\,\code{Ok}(\code{a}) \Rightarrow \MForce{\code{return}~\code{a}}
\end{align*}\end{footnotesize}
Where $\code{fail}$ is defined analogously to \code{return}, but using invoking the closest \code{.try} rather than \code{.kont} frame.
We verify in
\ifextended
\cref{sec:appendix:monad-laws}
\else
the extended version of the paper \cite{extended-version}
\fi
that this condition
ensures the monad laws hold.
This is of course a strong restriction, and essentially ensures the
implementor of a \code{Exn E A} computation only interacts with the
monad by calling \code{fail}, \code{return}, and \code{bind}.

\section{Monadic Blocks and the Algebra Translation}
\label{sec:mo-blocks-and-alg-trans}

While relative monads provide a nice interface for user-defined
effects implemented on the stack, their ergonomics leave a bit to be
desired, as writing programs using explicit calls to \code{bind} leads
to a CPS-like programming style.
In Haskell, this is solved using the \code{do} notation which allows
programming in a kind of embedded CBV sub-language.
Then a natural question is if we can adapt the Haskell approach and
develop a corresponding \code{do} notation for CBPV.
Of course, there is already a primitive \code{do} notation,
which works for one particular relative monad: $\BRet$. Then the
analog of Haskell's \code{do} notation would be a way to
\emph{overload} the built-in \code{do} to work with an arbitrary
relative monad.

However we cannot simply desugar all uses of \code{do} in CBPV to the
\code{bind} of a relative monad. In the \code{bind} for a relative
monad \code{T}, the continuation for the bind has result type \code{T
  A}, but the built-in \code{do} allows for the result type of the
continuation to be an arbitrary computation type \code{B}.
This means for example that it is trivial to embed a computation with
type $\BRet{A}$ in a computation of a user-defined monad type: 
we just need to use the built-in $\code{do}$ with a continuation of
our monadic type.

At first look, then, it seems like an overloaded \code{do} would need
to be much more restrictive than the built-in one, only allowing
continuations with the same monadic type. Miraculously, this is not
the case! Any relative monad can be used to interpret code written
with \emph{unrestricted} CBPV \code{do} notation, using the primitive
\code{bind} when the continuation is of the form \code{T A}, but using
a different, type-dependent operation when the continuation has a
different type \code{B}. This type-dependent operation is called an
\emph{algebra} of the monad. We show in this section that almost all
type constructors in CBPV canonically induce an associated algebra for
an arbitrary relative monad, allowing for a type-directed translation
of unrestricted CBPV \code{do} notation. The exception are the
$\BRet$, which is instead reinterpreted in our algebra translation as
the user defined monad $T$, and quantification $\forall$, where such
quantification must be translated to additionally take in some
associated algebra or algebra-transformer for the quantified type.

\subsection{Algebras of Relative Monads}
\label{sec:algebras-of-relative-monads}

Intuitively, an algebra is a version of \code{.bind} that works with
continuations whose output computation type is not necessarily of the
form \code{T~A}. In terms of stacks, an algebra for a type \code{B} is
a way of composing \code{B} stacks with \code{T~A} continuations.
Since algebras generalize bind, they have analogous laws. The only one
that does not have an analog is the right unit law, which is specific
to monadic result type.
\begin{definition}
  Algebras for a monad are given by the type
  \begin{footnotesize}
    \[\BAlg~\defeq~\SLam{(\code{T}: \VTy \to \CTy)~(\code{B}: \CTy)}{\BForall{\code{A}: \VTy}{\AThunk{(\code{T}~\code{A})} \to \AThunk{(\code{A} \to \code{B})} \to \code{B}} }\]
  \end{footnotesize}
  An algebra $\code{alg}: \AThunk{\BAlg~\code{T}~\code{B}}$ is
    \emph{lawful} if it satisfies the following equational principles:
  \begin{enumerate}
  \item \textbf{left unit law:} \begin{footnotesize}\[
    \MForce{\code{alg}}~\code{A}~\VThunk{\MForce{\code{m}~\code{.return}}~\code{A}~\code{a}}~\code{f}
    \equiv \MForce{\code{f}}~\code{a}
  \]\end{footnotesize} for any $\code{a}: \code{A}$ and $\code{f}: \AThunk{(\code{A} \to \code{B})}$,
  \item \textbf{associativity law:} \begin{footnotesize}\[
    \MForce{\code{alg}}~\code{A'}~\VThunk{\MForce{\code{.bind}}~\code{A}~\code{A'}~\code{t}~\code{f}}~g
    \equiv \MForce{\code{alg}}~\code{A}~\code{t}~\VThunk{\MTmLam{\code{x}{\MForce{\code{alg}}~\code{A'}~\VThunk{\MForce{\code{f}}~\code{x}}~\code{g}}}
  \]\end{footnotesize} for any $\code{f}: \AThunk{(\code{A} \to \code{T}~{\code{A'}})}$, $\code{g}: \AThunk\code{A'} \to \code{B})}$, and $\code{t}: \AThunk{(\code{T}~\code{A})}$,
  \item \textbf{linearity law:} \begin{footnotesize}\[
    \MBind{\code{t}}{\MForce{\code{tt}}}{\MForce{\code{alg}}~\code{A}~\code{t}} \equiv \MForce{\code{alg}}~\code{A}~\VThunk{\MBind{\code{t}}{\MForce{\code{tt}}}{\MForce{\code{t}}}}
  \]\end{footnotesize} for any $\code{tt}: \AThunk{(\BReturn{(\AThunk{(\code{T}~\code{A})})})}$.
  \end{enumerate}
\end{definition}

Then as expected, the built-in CBPV \code{do} notation can be used to
define a built-in $\BRet$-algebra for \emph{all} computation types
\code{B} structure \code{Algebra Ret B}:
\begin{footnotesize}
\begin{align*}
\code{alg\_ret}:~\AThunk{(\BForall{\code{B}} \BAlgebra{\BRet}{\code{B}})}~\defeq~\{~\MTyLam{\code{B}}{\MTyLam{\code{Z}}{\MTmLam{\code{tz}~\code{k}}{
  \keyword{do}~\code{z}~\leftarrow~!~\code{tz};~!~\code{k}~\code{z}
}}}~\}
\end{align*}
\end{footnotesize}
Additionally, for every monad \code{T} we can define the free algebra \code{Algebra T (T A)} for any \code{A} using the \code{.bind} of
the monad.
\begin{footnotesize}
\begin{align*}
  \code{alg\_mo}:~\AThunk{(\BForall{\code{T}} \code{Monad}~\code{T} \to \BForall{\code{X}}\BAlgebra{\code{T}}{\code{X}})}~&\defeq~\{~\MTyLam{\code{T}}{\MTmLam{\code{m}}{\MTyLam{\code{X}}{ \MTyLam{\code{Z}}{\MTmLam{\code{tz}~\code{k}}{\MForce{\code{m}}~\code{.bind}~\code{tz}~\code{k}}}}
    }}~\}
\end{align*}
\end{footnotesize}
What's more, each of the type constructors comes
with a canonical way of lifting algebra structures, that is, if the
arguments are algebras, then the constructed type is as well. Here is
the constructor for $\to$, the others are included in the
\ifextended{appendix}\else{extended version of the paper \cite{extended-version}}\fi.
\begin{footnotesize}
\begin{align*}
  \code{alg\_arrow}:~&\AThunk{(\BForall{\code{T}~\code{A}~\code{B}} \BAlgebra{\code{T}}{\code{B}} \to \BAlgebra{\code{T}}{(\code{A} \to \code{B})})}~\defeq\\
  &\{~\MTyLam{\code{T}}\MTyLam{\code{T}~\code{A}~\code{B}}{\MTmLam{(\code{alg}: \AThunk{(\BAlgebra{\code{T}}{\code{B}})})}{ \\
      &\qquad\quad\MTyLam{\code{Z}}{\MTmLam{\code{tz}~\code{f}}{
        \MTmLam{\code{a}}{\MForce{\code{alg}}~\code{tz}~\VThunk{~\MTmLam{\code{z}}{\MForce{\code{f}}~\code{z}~\code{a}~}}}
      }}
    }}~\}
\end{align*}
\end{footnotesize}
This algebra works by popping the \code{A} value off of the stack, and
then invoking the algebra for \code{B} with a continuation that
passes \code{A} to the original continuation. The algebra lifting
for $\&,\nu$ are simple to define and provided in the
\ifextended{appendix}\else{extended version}\fi
. This means complex types like $A \to^+ B$ we defined earlier can have algebra structures derived in a type-directed manner.

Finally, consider the computation type connective $\forall$. It is not the case
that arbitrary $\forall$ quantified types can be given algebra
structure, for example, how would we construct an algebra for
$\BForall{\code{R}}{\AThunk{\code{R}} \to \code{R}}$? Instead what we can support is quantifiers that
quantify over \emph{both} a type, \emph{and} an algebra structure. For
example $\BForall{\code{R}}{\BAlgebra{\code{T}}{\code{R}} \to \AThunk{\code{R}} \to \code{R}}$, in which
case the algebra structure can be defined, using the passed in algebra
structure for $\code{R}$. This process becomes more complicated if $\code{R}$ is not
a computation type, but a higher-kinded type, in which case we need to
pass in not an algebra for $\code{R}$, but some kind of ``algebra
constructor'' whose structure is derived from the kind of $\code{R}$. We make
this precise in the next section.

\subsection{Monadic Blocks}
\label{sec:mo-blocks}

Monadic blocks are introduced in \cref{fig:syntax_monadic_block}
as a new term constructor for computations. Given a \emph{closed}
computation $M$, we can construct a computation $\MMonadic{M}$ which
is like the original $M$ except that it is \emph{polymorphic} in the
choice of underlying monad. That is, everywhere the original $M$ used
$\BRet$, $\MMonadic{M}$ will use the passed in monad
\code{T}. Additionally, as discussed in the previous section,
everywhere $M$ quantified over types, $\MMonadic{M}$ must quantify
over \emph{both} a type and a corresponding \emph{structure}. This
requires that the output of $M$ be translated to $\floor{M}$, which we
will describe shortly.

\begin{figure}
  \begin{footnotesize}
    \begin{mathpar}
  \judgment{
    \CtxtEmpty;\CtxtEmpty \vdash M : B
  }{
    \Delta;\Gamma \vdash \MMonadic{M} : \BForall{(\code{T}: \VTy \to \CTy)}{\AThunk{(\BMonad{\code{T}})} \to \flCar{B}}
  }{TyMoBlock}
\end{mathpar}
  \end{footnotesize}
\caption{Monadic Blocks}
\label{fig:syntax_monadic_block}
\end{figure}

It may seem overly limiting to require $M$ to be a \emph{closed}
computation. However some limitations on $M$ are necessary: $M$ cannot
for instance use arbitrary computation type variables, as it would
require that these support algebras of the given monad. In practice,
we can work around this restriction by making $M$ polymorphic in its
free type variables. This would make the translation type for $M$ also
be polymorphic, but quantifying over the required algebra structures
as well.

Next, we define the type $\floor{B}$ of a monadic block in
Figure~\ref{fig:signature_carrier_translations}. We call this the
``carrier translation'' as it takes every type in CBPV to the type of
``carriers'' for an algebraic structure that will be constructed
later.
We use dependently typed syntax to give a specification for the
well-typedness property of this translation, using curly braces in the
style of Agda to denote inferrable arguments.
The carrier translation takes in a type $\Delta \vdash S : K$ and
returns back a type of the same kind, but requiring that a type
constructor $\code{T}$ (corresponding to the monad) additionally be in scope.
For most types, the translation is simply homomorphic and so we elide
these for now and include them in the \ifextended{appendix}\else{extended version}\fi.
We show the interesting cases: $\BRet$ is translated to the provided
monad $\code{T}$, and $\forall,\exists$ are translated to require or provide
respectively a \emph{structure} $\SigP{K}{\code{X}}$ for the quantified type $\code{X}$.
The type $\SigP{K}{S}$ is defined next, it defines for each kind $K$ a
computation type which is the ``signature'' type for a structure that
will be required for all types $S$ of that kind, again with a monad $\code{T}$
assumed to be in scope.
For computation types, that structure is that of an algebra of the
monad. For value types, the structure is trivial, so we use the trivial computation type
$\BTop$. For function kinds, we need a corresponding ``structure
constructor'': given any structure on the input, we require structure
on the output.
Additionally, we extend the notion of signature from kinds to type
environments: given a type environment $\Delta$, we return a value
environment $\SigD{\Delta}$ of thunks of structures satisfying the
signatures for the corresponding variables in $\Delta$.

%% at first glance,
%% because it only allows a closed computation to be wrapped in itself,
%% without access to all the prelude utilities in the language. However,
%% we'll show that it's possible to ``import'' types and terms from
%% outside the monadic block in
%% \cref{sec:importing-external-types-and-terms-into-monadic-blocks}.

\begin{figure}
  \begin{footnotesize}
\begin{align*}
  \floor{\,\cdot\,}~&:~\forall~\Implicit{K}~\Implicit{\Delta} \to (\Delta \vdash \_ : K) \to (\code{T}: \VTy \to \CTy, \Delta \vdash \_ : K) \\
  \floor{\BRet}~&:=~\code{T} \\
  \floor{\BForall{\code{X}:K}{B}}~&:=~\BForall{\code{X}:K}{\AThunk{(\SigP{K}{\code{X}})} \to \floor{B}} \\
  \floor{\AExists{\code{X}:K}{A}}~&:=~\AExists{\code{X}:K}{\AThunk{(\SigP{K}{\code{X}})} \times \floor{A}} \\
  \dots
\end{align*}
\begin{align*}
  \textrm{Sig}~&:~\forall~\Implicit{K}~\Implicit{\Delta} \to (\code{T}: \VTy \to \CTy, \Delta \vdash K) \to (\code{T}: \VTy \to \CTy, \Delta \vdash \CTy) \\
  \SigP{\CTy}{B}~&:=~{\BAlgebra{\code{T}}{B}} \\
  \SigP{\VTy}{A}~&:=~{\BTop} \\
  \SigP{K_0 \to K}{S}~&:=~{\BForall{\code{X}:K_0}{\AThunk{(\SigP{K_0}{\code{X}})} \to \SigP{K}{S~\code{X}}}} \\
  \\
  \textrm{Sig}~&:~(\Delta : \textrm{TEnv}) \to \Delta \vdash \textrm{VEnv} \\
  \SigD{\,\cdot\,}~&:=~\cdot \\
  \SigD{\Delta, \code{X}:K}~&:=~\SigD{\Delta}, \code{str}_{\code{X}}: \AThunk{(\SigP{K}{\code{X}})}
\end{align*}
  \end{footnotesize}
\caption{Signature Translation and Carrier Translation}
\label{fig:signature_carrier_translations}
\end{figure}

\subsection{The Algebra Translation}
\label{sec:algebra-translation}

To implement monadic blocks, we define a translation from CBPV to CBPV
that translates all monadic blocks to ordinary CBPV code. This
proceeds in two phases: first there is a translation $\floor{M}$ that
translates a closed program that does \emph{not} use monadic blocks
into one that is polymorphic in the underlying monad. Second, we have
a translation $\brac{M}$ that uses $\floor{-}$ to translate away all
uses of monadic blocks. This is defined homomorphically on all
structures except monadic blocks which are translated as
\begin{footnotesize}
  \[ \brac{\MMonadic{M}}~:=~\MTyLam{(\code{T}: \VTy \to \CTy)}{ \MTmLam{(\code{m}: \AThunk{(\BMonad{T})})}{\floor{\brac{M}}}} \]
\end{footnotesize}

The term translation for values and computations is given in
Figure~\ref{fig:term_translations}.  The type-preservation property of
the translation is that given an \emph{open} value/computation, we
construct a value/computation whose output type is given by the
carrier translation, but under a context extended with a monad
\code{T} as well as structures for all of the free type variables
$\SigD{\Delta}$, with the free value variables translated according to
the carrier translation $\floor{\Gamma}$.
The term translation elaborates \keyword{ret} to the return of the
provided monad and elaborates \keyword{do}-bindings to use the algebra
constructed by the structure translation.
The $\forall$ and $\exists$ cases are translated to be keeping the
thunked structure around when instantiating the type variable $\code{X}$,
which is bound to $\code{str}_{\code{X}}$ whenever the $\code{X}$ is bound.

\begin{figure}
  \begin{footnotesize}
    \begin{align*}
  \floor{\,\cdot\,}~&:~\forall~\Implicit{A}~\Implicit{\Delta}~\Implicit{\Gamma} \\
                    &\,\,\,\,\to (~\Delta; \Gamma \vdash A~) \\
                    &\,\,\,\,\to (~T: \VTy \to \CTy, \Delta \\
                    &\qquad\,\,;~mo: \AThunk{(\BMonad{T})}, \SigD{\Delta}, \floor{\Gamma} \vdash \floor{A}~) \\
  \floor{x}~&:=~x \\
  \floor{\VPack{S}{V}}~&:=~\VPack{\floor{S}}{\VPair{\VThunk{\StrP{S}}}{\floor{V}}} \\
  \dots \\
  \\
  \floor{\,\cdot\,}~&:~\forall~\Implicit{B}~\Implicit{\Delta}~\Implicit{\Gamma} \\
                    &\,\,\,\,\to (~\Delta; \Gamma \vdash B~) \\
                    &\,\,\,\,\to (~\code{T}: \VTy \to \CTy, \Delta \\
                    &\qquad\,\,;~\code{m}: \AThunk{(\BMonad{T})}, \SigD{\Delta}, \floor{\Gamma} \vdash \floor{B}~) \\
  \floor{\MLetPack{\code{X}}{\code{x}}{V}{M}}~&:=~\MLetPack{\code{X}}{\code{p}}{\floor{V}}{\MLetPair{\code{str}_{\code{X}}}{\code{x}}{\code{p}}{\floor{M}}} \\
  \floor{\MReturn{V}}~&:=~\,\MForce{\code{m}}~\code{.return}~\floor{A}~\floor{V}
    \\ &\qquad\text{where}~\Delta;\Gamma \vdash V:A \\
  \floor{\MBind{\code{x}}{M_0}{M}}~&:=~\StrP{B}~\floor{A}~\VThunk{\floor{M_0}}~\VThunk{\MTmLam{\code{x}}{\floor{M}}}
    \\ &\qquad\text{where}~\Delta;\Gamma \vdash M_0:\BReturn{A}~\text{and}~\Delta;\Gamma \vdash M:B \\
  \floor{\MTyLam{\code{X}:K}{M}}~&:=~\MTyLam{\code{X}:K}{\MTmLam{\code{str}_{\code{X}}}\floor{M}} \\
  \floor{M~S}~&:=~\floor{M}~\floor{S}~\VThunk{\StrP{S}} \\
  \dots
\end{align*}
  \end{footnotesize}
\caption{Term Translations (Selected Cases)}
\label{fig:term_translations}
\end{figure}

The structure translation is given in
Figure~\ref{fig:structure_translation}.  The translation $\StrP{S}$
takes a type (constructor) $S:K$ and generates a computation of type
$\SigP{K}{\floor{S}}$ in the target language, scoped by the additional
monad type constructor $\code{T}$ and a monad instance $\code{m}$, as well as all
of the assumed structures for any free type variables. For value type
constructors, the resulting structure is defined trivially. Type
operators like $\AThk$ or $(\times)$ takes the type arguments and its
structure first before generating the algebra for the value or
computation type base cases.
As for the computation types, we use the lifting of algebras described
earlier in this section.  In our cases, the parameter $\code{tz} :
\AThunk{(\code{T}~\code{Z})}$ is the thunked monad value, and $\code{f}$ is the
continuation from $\code{Z}$ to the translated type.
The $\BForall{\code{X}:K}{B}$ case is a base computation type, so it starts
with parameters $\code{Z}$, $\code{tz}$ and $\code{f}$ as part of the algebra definition,
leaving a body of type $\floor{\BForall{\code{X}:K}{B}}$, which is
$\BForall{\code{X}:K}{{\AThunk(\SigP{K}{\code{X}})} \to \floor{B}}$. After
introducing the quantified type $\code{X}$ and its structure $\code{str}_{\code{X}}$ into
scope, we can now generate the algebra structure of $B$, pass in all
the arguments we received along the way, and use type application in
the continuation.  
The $\BNu{\code{Y}:K}{B}$ case is a bit more complicated, but the same idea
applies. We first introduce a fixed point $\code{str}$ for the generated
algebra structure, introduce the algebra parameters, and then
construct a term of the type $\floor{\BNu{\code{Y}:K}{B}}$ by using
\keyword{roll}. When we reach the recursive call $\StrP{B}$, we've
introduced both a abstract type variable $\code{Y}$ and a variable $\code{str}_{\code{Y}}$
for the the structure of $\code{Y}$. After substituting them with
$\BNu{\code{X}:K}{B}$ and $\code{str}_{\code{Y}}$ accordingly, we can pass in the arguments
as usual and finally \keyword{unroll} it in the continuation.  
At last we have the type variable, type constructor, and type
instantiation cases. The type variable case $\code{X}$ is a simple
lookup of $\code{str}_{\code{X}}$ in the value environment. In the
case for a type-level function $\lambda \code{X}. S$, the resulting
structure is a \emph{structure constructor}, taking in the type
variable $\code{X}$ and its associated structure
$\code{str}_{\code{X}}$ and producing a structure for $S$.  Then type
application case $S\,S_0$ supplies the structure of the argument type
$S_0$ to the structure constructor $S$.

\begin{figure}
  \begin{footnotesize}
    \begin{align*}
  Str~&:~\forall~\Implicit{K}~\Implicit{\Delta} \to (\Delta \vdash S: K) \\
      &\,\,\,\,\to (\code{T}: \VTy \to \CTy, \Delta;~\code{m}: \AThunk{(\BMonad{\code{T}})}, \SigD{\Delta} \vdash \SigP{K}{\floor{S}}) \\
  \StrP{\AThk}~&:=~\MTyLam{\code{X}}{\MTmLam{\_}{\MTop}} \\
  \StrP{(\times)}~&:=~\MTyLam{\code{X}}\MTmLam{\_}\MTyLam{\code{Y}}\MTmLam{\_}\MTop \\
  % \StrP{(+_C)}~&:=~\multi{\MTyLam{X}\MTmLam{\_}}_{c \in C}~\MTop \\
  \dots \\
  \StrP{\BRet}~&:=~\MTyLam{\code{X}}{\MTmLam{\_}{\MTyLam{\code{Z}}{\MTmLam{\code{tz}~(\code{f}: \AThunk{(\code{Z} \to \code{T}~\code{X})})}{
    \\ &\qquad\,\,\,\,\,\,\MForce{\code{m}}~\code{.bind}~\code{Z}~\code{X}~\code{tz}~\code{f}
  }}}} \\
  \StrP{(\to)}~&:=~\MTyLam{\code{X}}{\MTmLam{\_}{\MTyLam{\code{Y}}{\MTmLam{\code{alg}_{\code{Y}}}{
    \MTyLam{\code{Z}}{\MTmLam{\code{tz}~(\code{f}: \AThunk{(\code{Z} \to \floor{\code{X}} \to \floor{\code{Y}})})}{
      \\ &\qquad\quad\MTmLam{(\code{x}: \floor{\code{X}})}{\MForce{\code{alg}_{\code{Y}}}~\code{Z}~\code{tz}~\VThunk{\MTmLam{\code{z}}{\MForce{\code{f}}~\code{z}~\code{x}}}}
    }}
  }}}} \\
  \StrP{(\&_D)}~&:=~\multi{\MTyLam{\code{X}_d}\MTmLam{\code{alg}_d}}_{d \in D}
    \MTyLam{\code{Z}}\MTmLam{\code{tz}~(\code{f}: \AThunk{(\code{Z} \to \BWith{.d:\floor{\code{X}_d}}{d \in D})})}
      \\ &\qquad\quad\MComatch
        {.d}
          {\MForce{\code{alg}_d}~\code{Z}~\code{tz}~\VThunk{\MTmLam{\code{z}}{\MForce{\code{f}}~\code{z}~.d}}}
        {d \in D}
  \\
  \StrP{\BForall{\code{X}:K}{B}}~&:=~\MTyLam{\code{Z}}{\MTmLam{\code{tz}~(\code{f}: \AThunk{(\code{Z} \to \BForall{\code{X}:K}{{\AThunk(\SigP{K}{\code{X}})} \to \floor{B}})})}{
    \\ &\qquad\quad\MTyLam{\code{X}:K}{\MTmLam{(\code{str}_{\code{X}}: \AThunk{(\SigP{K}{\code{X}})})}{ \StrP{B}~\code{Z}~\code{tz}~\VThunk{\MTmLam{\code{z}}{ \MForce{\code{f}}~\code{z}~\code{X}~\code{str}_{\code{X}} }}}}
  }} \\
  \StrP{\BNu{\code{Y}:K}{B}}~&:=~\MFix{(\code{str} : \AThunk{(\BAlgebra{\code{T}}{(\BNu{\code{Y}:K}{\floor{B}})})})}{
    \\ &\qquad\quad\MTyLam{\code{Z}}{\MTmLam{\code{tz}~(\code{f}:\AThunk{\code{Z} \to \BNu{\code{Y}:K}{\floor{B}}})}{
      \\ &\qquad\quad\quad\MRoll{
        \subst
          {\subst{\StrP{B}}{(\BNu{\code{Y}:K}{\floor{B}})}{\code{Y}}}
          {\code{str}}{\code{str}_{\code{Y}}}
          ~\code{Z}~\code{tz}~\VThunk{\MTmLam{\code{z}}{\MUnroll{\MForce{\code{f}}~\code{z}}}}
      }
    }
  }} \\
  \StrP{\code{X}}~&:=~\MForce{\code{str}_{\code{X}}} \\
  \StrP{\SLam{\code{X}:K}{S}}~&:=~\MTyLam{\code{X}}{\MTmLam{\code{str}_{\code{X}}}{\StrP{S}}} \\
  \StrP{S~S_0}~&:=~\StrP{S}~\floor{S_0}~\VThunk{\StrP{S_0}}
\end{align*}
  \end{footnotesize}
\caption{Structure Translation (Selected Cases)}
\label{fig:structure_translation}
\end{figure}

%% \subsection{Importing External Types and Terms into Monadic Blocks}
%% \label{sec:importing-external-types-and-terms-into-monadic-blocks}

%% We can implement ``import'' in the same way that we pass in the monad
%% type constructor and the monad instance to the monadic block - by
%% turning the inner computation of a monadic block to polymorphic
%% functions that accepts abstract types, structures, and values.

%% \yuchen{an example of monadplus of option (monoid instance)}
%% \yuchen{an example of exit}
%% \yuchen{an example of add that utilizes monad homomorphism}

\subsection{Deriving Relative Monad Transformers}
\label{sec:monad-transformers}

In Haskell, many monads have a corresponding monad \emph{transformer}
\cite{monad-transformers,jaskelioff-2009}
which provides a way to ``add'' effects to an arbitrary input
monad. A monad transformer is a type constructor $t : (* \to *) \to (* \to *)$
such that if $m$ is a monad then $t~m$ is a monad and further there is
a polymorphic lifting operation $m~a \to t~m~a$ that preserves return
and bind. We think of $t~m$ as the monad $m$ extended with the effects
from $t$. The lifting operation allows us to embed $m$ effects in the
transformed monad.

% Note: adding a line break according to the review
Most common monads in Haskell extend to monad
transformers. That is, for most monads $m$ there is a corresponding monad
transformer $t$ such that $m$ is isomorphic to $t~\textrm{Identity}$.
For instance, the error monad extends to the error
monad transformer \code{EitherT e m a = (m (Either e a))} and the
state monad extends to the state monad transformer \code{StateT s m a = s -> m (a, s)}.
However this is not the case for all monads.
Most well known is the list monad, which models a kind of
quantitative, ordered non-determinism. We may na\"ively attempt to
extend this to a monad transformer in a similar manner to the previous
examples $\code{ListT m a = m [a]}$, but this fails to satisfy the
monad laws unless we further assume that \code{m} is a commutative
monad.

We are naturally led then to consider relative monad transformers in
CBPV, as these would have similar compositionality benefits. We define
relative monad transformers in \cref{fig:mt_def}. A relative monad
transformer is a structure on a type constructor $\code{F}: (\VTy \to \CTy) \to \VTy \to
\CTy$ such that for any relative monad $\code{T}$, $\code{F}~\code{T}$ is a relative monad,
and we support a $\code{lift}$ function that embeds $\code{T}~\code{A}$ computations
into $\code{F}~\code{T}~\code{A}$ computations. Then we say a monad $\code{S}$ extends to a monad
transformer $\code{F}$ if $\code{S}$ is isomorphic to $\code{F}~\BReturn$: $\BReturn$ being
the analog of the identity monad in Haskell.

The process of deriving a relative monad transformer from a relative
monad requires less insight than in the case of ordinary monads. For
instance, the sum-based exception monad $\code{Exn}~\code{E}~\code{A} =
\BReturn{(\ACoProdLit{\code{Err}:\code{E},~\code{Ok}:\code{A}})}$ is easily transformed into a relative
monad transformer by replacing $\BReturn$ with an arbitrary input monad $\code{T}$:
$\code{ExnT}~\code{T}~\code{E}~\code{A} = \code{T}~(\ACoProdLit{\code{Err}:\code{E},~\code{Ok}:\code{A}})$ and similarly for a
non-continuation based state monad $\code{StateT}~\code{T}~\code{S}~\code{A} = \code{S} \to
\code{T}~(\AProd{\code{S}}{\code{A}})$. We notice that this replacement of $\BReturn$ with
an arbitrary relative monad is reminiscent of our algebra
translation. In fact, our algebra translation can automate this
process entirely!

\begin{figure}
  \centering
  \begin{footnotesize}
    \begin{align*}
    \BMoTrans~&\defeq~\SLam{(\code{F}: (\VTy \to \CTy) \to \VTy \to \CTy)}~
    \\
    &\qquad\quad\BForall{(\code{T}: \VTy \to \CTy)} \AThunk{(\BMonad{\code{T}})} \to
    \\
    &\qquad\qquad\quad\begin{aligned}
      &   \&\{\,\,\code{.monad}~:~\BMonad{(\code{F}~\code{T})} \\
      &\quad\,\,\,\code{.lift}~:~\BForall{(\code{A}: \VTy)}{\AThunk{(\code{T}~\code{A})} \to \code{F}~\code{T}~\code{A}}\,\,\}
    \end{aligned}
  \end{align*}
  \end{footnotesize}
  \caption{Relative Monad Transformers}
  \label{fig:mt_def}
\end{figure}

\begin{figure}
  \begin{footnotesize}
  \textrm{Assume} $\cdot;\cdot\vdash M : \BMo~S$
    \begin{align*}
  \code{motrans}: \BMoTrans(\SLam{\code{T}}~\lfloor S \rfloor)~&\defeq~\MTyLam{\code{T}}\MTmLam{\code{m}}
    \MForce{\code{removeTriv}}~\{~\mathcal{M}( \\
      &~\qquad\keyword{let}~{\code{mo}_{\code{f}}} = \VThunk{~M~} ~\keyword{in}~\keyword{comatch} \\
      &~\qquad\quad|~\code{.monad} \Rightarrow \MForce{\code{mo}_{\code{f}}} \\
      &~\qquad\quad|~\code{.lift} \Rightarrow \MTyLam{\code{Z}}\MTmLam{\code{tz}}\MBind{\code{z}}{\MForce{\code{tz}}}{\MForce{\code{mo}_{\code{f}}}~\code{.return}~\code{Z}~\code{z}} \\
    &\quad\,\,\,)~\code{T}~\code{mo}~\}
  \\
  \code{removeTriv}~&\defeq~\MTmLam{\code{m}}
    \keyword{comatch} \\
    &~\quad\quad|~\code{.monad} \Rightarrow \keyword{comatch} \\
    &~\quad\quad\quad|~\code{.return} \Rightarrow \MTyLam{\code{A}}\MForce{\code{m}}~\code{.monad}~\code{.return}~\code{A}~\code{triv} \\
    &~\quad\quad\quad|~\code{.bind}\Rightarrow \MTyLam{\code{A}~\code{A'}} \MForce{\code{m}}~\code{.monad}~\code{.bind}~\code{A}~\code{triv}~\code{A'}~\code{triv} \\
    &~\quad\quad|~\code{.lift} \Rightarrow \MTyLam{\code{A}} \MForce{\code{m}}~\code{.lift}~\code{A}~\code{triv}
  \\
  \code{triv}~&\defeq~\VThunk{~\MTop~}
\end{align*}
  \end{footnotesize}
\caption{Deriving Monad Transformers from Monads}
\label{fig:deriving-monad-transformers}
\end{figure}

In \cref{fig:deriving-monad-transformers} we show how to use monadic
blocks to automatically derive monad transformers from a closed
implementation of a monad $M : \code{Monad}~S$. The corresponding type
constructor is given using the algebra translation $\lfloor S
\rfloor$. To implement the monadic and lift operations, we enter a
monadic block for the provided monad $T$ that is to be transformed.
Then inside the monadic block, we embed the monad implementation $M$.
This lift implementation inside the monadic block has the type
$\BForall{\code{A}}.~\AThk{(\BReturn{\code{A}})} \to S~\code{A}$, which is easily implemented
using the built-in $\keyword{do}$ notation for $\BRet$.  When we view
these two operations outside the monadic block, they almost correspond
exactly to the $\code{.monad}$ and $\code{.lift}$ type signatures,
except that quantification over value types additionally takes in
trivial structures. For this reason, to get something that meets the
monad transformer type signature exactly, we $\eta$-expand the
monad/lift definitions and fill in these trivial structures.

This mechanical process produces the expected monad transformer
definitions for $\code{ExnT}$ and $\code{StateT}$, but what about
continuation-based monads that don't directly mention $\BRet$ in
their definition? These cases involve quantification over a result
computation type, and in the monad transformer this is generalized to
take in an \emph{algebra} of the provided monad. For instance, for
\code{ExnK}, we derive the type of the monad transformer to be:
$\code{ExnKT}~\code{E}~\code{M}~\code{A} = \BForall{\code{R}} \BAlg~\code{M}~\code{R} \to \AThk{(\code{E} \to
  \code{R})} \to \AThk{(\code{A} \to \code{R})} \to \code{R}$.

How does this relate to the fact that Haskell's list monad does not
extend to a monad transformer? In CBPV, we can define a type of lists
of values using recursive value types, but we will be unable to define
a relative monad $\code{ListM}~\code{A} = \BReturn{(\code{List}~\code{A})}$ without
assuming that $\BRet$ satisfies commutativity. That is, we
encounter the same problem as in Haskell, but because all relative
monads extend to monad transformers, we already encounter the issue
when defining a relative monad.

\section{Fundamental Theorem of CBPV Relative Monads}
\label{sec:fund-thm}

In this section, we connect our development of relative monads in CBPV
to the category-theoretic notion of relative monad, and prove the
semantic analog of the algebra translation of the previous section, a
theorem we call the \emph{fundamental theorem of CBPV relative
monads}. While we have presented the syntactic translation first, we
note that we arrived at the mathematical version first and then
adapted this to our syntactic calculus.

In order to reduce the mathematical complexity of this section, we
will focus only on the computation type constructors in models of
simply typed CBPV, that is without considering polymorphism ($\forall$),
recursive types ($\nu$) or term-level fixed points. We expect more general
versions of the theorem are possible, such as a version for
extensional dependently typed CBPV models, but leave this to future work.

\subsection{Models of Call-by-push-value}

%% Next, a strong CBPV model is a model of simply typed CBPV and is very
%% similar to Levy's original notion used in CBPV adjunction models. The
%% formulation here in terms of presheaf-enrichment is due to
%% \citet{cfm-2016}.
%% %
%% Strong models generalize unary models because we can reindex $\mathcal
%% E$ to be an ordinary category and define $J\,A\,B = \mathcal C_A(B)$.
%% %
%% Dependently typed CBPV models generalize strong CBPV models by
%% allowing value and computation types to be dependent on value types.
%% %
%% The value types and pure morphisms are essentially an ordinary model
%% of dependent type theory, and then computation types, computations
%% stacks are internal to presheaves, which intuitively means they are
%% indexed by a context and contain an action of substitution which
%% commutes with stack composition.
%% %

For simplicity we work with a very simple notion of a CBPV model where
value types are interpreted as sets. We discuss the relationship to other
models in Section~\ref{sec:discussion}. We start by defining a model
for only the \emph{judgments} (value/computation types and terms) of
CBPV without any of the type connectives.
\begin{definition}
  A judgmental CBPV model consists of
  \begin{enumerate}
  \item A universe of value objects, that is a (large) set $\mathcal V_0$ and a
    function $\el : \mathcal V_0 \to \Set$.
  \item A category $\mathcal E$ of computation objects and linear
    morphisms
  \item A functor $\mathcal C : \mathcal E \to \Set$ of computations.
  \end{enumerate}

  For any such model, we define a category $\mathcal V$ to have
  $\mathcal V_0$ as objects and as morphisms $f \in \mathcal V(A,A')$
  just functions $f : \el(A) \to \el(A')$. Further we define a
  profunctor $J : \mathcal V^{op} \times \mathcal E \to \Set$ by
  $J(A,B) = \el(A) \to \mathcal C(B)$. We will notate the functorial
  actions of this profunctor by composition $\circ$.
\end{definition}

A judgmental CBPV model can model the basic notions of our CBPV
calculus, but few of the type constructors. Value types $A$ are
interpreted as value objects $A \in \mathcal V_0$, value environments
$\Gamma = x_1:A_1,\ldots$ are interpreted as the Cartesian product
$\el(\Gamma) = \el(A_1) \times \cdots $ with the empty context
interpreted as a one-element set. Computation types are interpreted as
computation objects $B \in \mathcal E_0$. Then values $\Gamma \vdash V
: A$ are modeled as functions $V : \el(\Gamma) \to \el(A)$, and
computations $\Gamma \vdash M : B$ are interpreted as functions
$\el(\Gamma) \to \mathcal C(B)$.

One difference between the syntax and semantics is that the semantics
has primitive notions of pure morphism (function between value
objects) and linear morphism (morphism between computation objects),
whereas the syntax has a more restrictive notion of values and we only
give a primitive syntax for stacks in the operational semantics. The correct
interpretation is that the morphisms between value objects $\mathcal
V(A,A')$ correspond in the syntax to \emph{thunkable} terms $x:\AMV
\vdash M : \BReturn{\AMV'}$ and similarly that computation morphisms
$\mathcal E(B,B')$ correspond to \emph{linear} terms $z:\AThunk{\BMV}
\vdash M : \BMV'$, as described in \cite{munch-linearity}. In our
semantics, we will work with elements of $\mathcal E(B,B')$ where in
our syntax we worked with a linearity restriction.

In order to model the type and term constructors of CBPV, a judgmental
model must come with corresponding object and morphism constructors. A
convenient feature of our choice of model is that the function $\to$
and product $\&$ computation types can both be modeled as products:
%% most type constructors in the syntax can be expressed as either
%% coproducts or products:
\begin{definition}
  %% \begin{enumerate}
  %% \item Let $I$ be a set and $A : I \to \mathcal V_0$ a family of
  %%   value objects. A coproduct value object is an object $\sum_{i \in
  %%     I} A(i)$ with a natural isomorphism $\mathcal V(\sum_{i \in I}
  %%   A(i), A') \cong \prod_{i \in I} \mathcal V(A(i), A')$.

  %% \item
    Let $I$ be a set and $B : I \to \mathcal E_0$ a family of
    computation objects. A product computation object is the categorical product of the family, i.e., it is a computation object
    $\with_{i : I} B(i)$ with a natural isomorphism $\mathcal E(B', \with_{i\in I} B(i))
    \cong \prod_{i \in I} \mathcal E(B', B(i))$.
  %% \end{enumerate}
\end{definition}
Then we define when a model of CBPV supports $\AThk$, $\BRet$, $\&$ and $\to$ as follows:
\begin{definition}
  \label{def:cbpv-model}
  A CBPV model is a judgmental CBPV model with the following
  structures:
  \begin{enumerate}
  \item (thunk objects) a functor $\AThk : \mathcal E \to \mathcal V$ and
    a natural isomorphism $\mathcal V(A,\AThk~B) \cong \mathcal J(A,B)$
  \item (returner objects) a functor $\BRet : \mathcal V \to \mathcal E$
    and a natural isomorphism $\mathcal E(\BRet~A,B) \cong \mathcal
    J(A,B)$.
  %% \item (nullary, binary value coproducts) All coproducts value
  %%   objects indexed by $0$ and $2$.
  %% \item (nullary value product) a value object $1$ with $\el(1)$ having one element.
  %% \item (binary value products) For any $A, A' \in \mathcal V_0$, a
  %%   value coproduct of the constant family $\lambda _. A' : \el(A) \to
  %%   \mathcal V_0$.
  \item (nullary, binary computation products) All product computation
    objects indexed by $0$ and $2$.
  \item (computation function types), For any $A \in \mathcal V_0, B
    \in \mathcal E_0$ a computation product of the constant family
    $K(B) : \el(A) \to \mathcal E_0$.
  \end{enumerate}
\end{definition}

By composing natural isomorphisms we get $\mathcal E(\BRet~A, B) \cong
\mathcal V(A, \AThk~B)$, that is, $\AThk$ is a right adjoint functor to $\BRet$.

\subsection{Relative Monads, Algebras and the Fundamental Theorem}

Next we recount the standard definition of a relative monad in
category theory \cite{acu-2010}.
\begin{definition}
  Let $G : \mathcal X \to \mathcal Y$ be a functor. A monad
  relative to $G$ consists of
  \begin{enumerate}
  \item A function $T : \mathcal X_0 \to \mathcal Y_0$
  \item For each $x \in \mathcal X_0$, a \emph{unit} $\eta_x : \mathcal Y(Gx, Tx)$
  \item For each morphism $f \in \mathcal Y(Gx,Tx')$, an \emph{extension} $f^\dag \in \mathcal Y(Tx,Tx')$
  \item Satisfying $\eta_x^\dag = \id$, $f^\dag \circ \eta = f$ and $(f^\dag \circ g)^\dag = f^\dag \circ g^\dag$ for all $f,g$
  \end{enumerate}
\end{definition}

How is this related to our syntactic definition in
Section~\ref{sec:relative-monads-in-zydeco}? Our syntactic definition
corresponds to the following semantic definition:
\begin{definition}
  A \emph{CBPV relative monad} consists of
  \begin{enumerate}
  \item A function $T : \mathcal V_0 \to \mathcal E_0$
  \item For each $A \in \mathcal V_0$, an element $\eta_A : J(A, T A)$
  \item For each $A,A'\in \mathcal V_0$ a function $-^\dag : J(A, T
    A') \to \mathcal E(T A, T A')$
  \item Satisfying $\eta^\dag = \id$, $j^\dag \circ \eta = j$ and
    $(j^\dag \circ k)^\dag = j^\dag \circ k^\dag$ for all $j,k$.
  \end{enumerate}
\end{definition}
This has a fairly direct correspondence to our syntactic
notion. The syntactic version of return denotes an element of $\AThk(A \to T A')$. But this is equivalent to the unit of a relative monad because $\AThk$ preserves powering by $\el(A)$ because it is a right adjoint:
\begin{footnotesize}
  \[ \AThk(\with_{\_:\el(A)} T A') \cong \el(A)\to \el(\AThk(T A')) \cong \mathcal V(A, \AThk(T~A')) \cong \mathcal J(A,T A')\]
\end{footnotesize}
The $-^\dag$ operation then corresponds to the bind operation with
arguments reordered, with the linearity assumption of the bind
operation instead encoded by requiring that the $-^\dag$ operation
construct a morphism in $\mathcal E$. Note that because we are working
with a model where $\mathcal V$ is a subcategory of sets, this notion
of monads is automatically ``strong'' in that we can define a version of
$-^\dag$ that works in an arbitrary context $(\el(\Gamma) \to \mathcal
J(A, T A')) \to (\el(\Gamma) \to \mathcal E(T A, T A')$ as simply
$f^\dag(\gamma) = (f(\gamma))^\dag$.

How then does a CBPV relative monad relate to the standard notion? The
definitions seem quite similar, the only difference is that a CBPV
relative monad uses the profunctor $J$ whereas the standard notion
uses morphisms out of $Gx$. This formulation of relative monads in terms of a profunctor
has been identified in prior work \cite{arkor2024formal,levy-relative-monad}. Then we need only observe that by \cref{def:cbpv-model} in any CBPV model we are provided
natural isomorphism $J(A,B) \cong \mathcal E(\BReturn{A},B)$ and then the
following is a straightforward consequence.
\begin{lemma}
  A CBPV relative monad is equivalent to a monad relative to $\BRet$
\end{lemma}

This equivalence allows us to adapt standard notions of algebra and
homomorphism to the CBPV setting.
\begin{definition}
  An \emph{algebra} for a CBPV monad $T$ consists of
  \begin{enumerate}
  \item A carrier object $B \in \mathcal E_0$
  \item For every $A$ a function $-^{\dag_B} : J(A, B) \to \mathcal E(T A,B)$
  \item Satisfying $j^{\dag_B} \circ \eta = j$ and $(j^{\dag_B} \circ k)^{\dag_B} = j^{\dag_B} \circ k^\dag$
  \end{enumerate}

  Given two algebras $(B,-^{\dag_B})$ and $(B',-^{\dag_B'})$, an
  algebra \emph{homomorphism} is a morphism $\phi : \mathcal E(B,B')$
  that preserves the extension operation $\phi(j^{\dag_B}) =
  \phi(j)^{\dag_B'}$. This defines a category $\Alg(T)$ with a
  forgetful functor $U : \Alg(T) \to \mathcal E$ that takes the
  carrier object of an algebra and the underling morphism of a
  homomorphism.
\end{definition}

Then the fundamental theorem of relative monads says that this
category of algebras can replace $\mathcal E$ and serve as a new model
of CBPV. To prove this we need only show that this category of
algebras can model the connectives of CBPV. Modeling the connectives
on value objects comes essentially for free as we will use the same
notion of value object. For the returner object, we use the \emph{free
algebra}, which is the semantic analog of the fact that every monad,
when instantiated is an algebra for itself:
\begin{definition}
  For a value object $A$, the free algebra is the $T$-algebra on $TA$
  defined by the $-^\dag$ of the monad itself. This has the property
  that algebra homomorphisms $\Alg(T)((TA,-^\dag), (B,-^{\dag_B}))$
  are naturally isomorphic to morphisms $J(A, B)$.
\end{definition}

The remaining computation object constructions, function objects and
nullary and binary products, are all given by computation products,
and the fact that $\Alg(T)$ has these products follows from the
following standard property of algebras of relative monads:
\begin{lemma}
  \label{lem:limits-of-algebras}
  The forgetful functor $U : \Alg(T) \to \mathcal E$ creates limits.
\end{lemma}

\begin{theorem}[Fundamental Theorem of CBPV Relative Monads]
  Given a CBPV relative monad $T$, we construct a new CBPV model with
  \begin{enumerate}
  \item $\mathcal V$ as in the original model, $\mathcal E = \Alg(T)$
    and $\mathcal C = U \circ \mathcal C$
  \item $\AThk = \AThk \circ U$, other value object connectives as in the original model
  \item $\BRet$ given by the free algebra, other computation connectives given by \cref{lem:limits-of-algebras}
  \end{enumerate}
\end{theorem}
Our algebra translation is essentially a syntactic version of this
construction: interpreting the syntax of CBPV into a new model where
computation types carry an algebra structure.

\section{Discussion and Future Work}
\label{sec:discussion}

\paragraph{Programming with CBPV}

Our polymorphic CBPV language and semantics is based on Levy's
original CBPV syntax and stack-machine semantics
\cite{levy01:phd}. \citet{beyond-polarity} further
discussed how different calling conventions can be modeled as
conservative extensions to Levy's CBPV through the shifts between
positive and negative types, which is similar to our use of CBPV for
encoding stack representations.

\paragraph{Categorical models of CBPV}
\label{sec:models-of-cbpv}

Our notion of CBPV model, in which values form a subcategory of sets
is seemingly much more restrictive than Levy's original formulation of
adjunction models or Curien, Fiore and Munch-Maccagnoni's (CFMM)
simplification\cite{levy-stack-2005, cfm-2016}. CFMM model
value types in an arbitrary Cartesian category, and model computation
types in a category \emph{enriched} in presheaves over the
category of values. We can view a CFMM model as one of our
models by interpreting ``set'' using the \emph{internal language} of
presheaves over value types. Then the value types do form a universe
(the representable presheaves), and the category of computation types
and linear morphisms can be viewed as an ordinary category in
this internal language. Then as long as we stick to constructive
arguments (which we have done in this work), any mathematics that we
do can be interpreted in presheaf categories. We refer interested
readers to existing resources on the interpretation of mathematics in
the internal language of presheaf categories \cite{maietti-2005}.

%% There are many existing approaches in the literature using indexed
%% categories \cite{levy-stack-2005}, categories enriched in presheaves
%% \cite{cfm-2016} and fibered categories \cite{ahmanghaniplotkin16}.

%% This notion is related to the one given by Ahman in the same way that
%% comprehension categories are related to natural models in the
%% semantics of dependent type
%% theory\cite{ahmanghaniplotkin16,awodey-2018}. It is also quite
%% similar to the theory of dependent CBPV used in the work on Calf,
%% except that the Calf work does not include a stack judgment
%% \cite{nsgh-2022}.

%% We
%% will work with a model that is at first glance less general where we
%% assume the category of value types and values is a subcategory of the
%% category of sets. This has the effect of removing the usual need to
%% encode a strength restriction on our adjunction, as adjunctions with a
%% subcategory of $\Set$ are always strong. However as long as all of our
%% proofs are intuitionistic, this can be justified by saying that we are
%% working internally to the category of presheaves on value types.

\paragraph{Compilation}
\label{sec:verified-compilation}

In this work, we have used CBPV as a setting that captures the essence
of stack-manipulation. In future work, we aim to study how to take
advantage of these mechanisms for abstraction over effects in the
setting of compiler intermediate representations. Several works have
related CBPV and similar polarized calculi to compiler intermediate
representations \cite{gmrz-2018,downen16,call-by-unboxed-value}. The
computation type structure bears similarities to the stack typing of
Stack-based Typed Assembly Language (STAL) \cite{mcgw-2003} and their
use of stack types to describe the implementation of exceptions was a
major influence on this work. We go beyond their double-barreled
continuation example to show that \emph{recursive} computation types
(in STAL, stack types) can be used to encode stack-walking
implementations as well. We aim to investigate if the relative monad
abstraction and relative monad transformers that we have identified
could be used to develop compositional methods for implementation and
verification of stack-based implementations of effects.

\paragraph{Alternative Implementations of the Algebra Translation}
\label{sec:implementing-algebras}

The algebra translation for monadic blocks is carefully crafted to
work exactly with the simple $F_\omega$-style of CBPV calculus that we
have presented here, but we discuss other programming language
features that would allow for different implementations.
If our language supported \emph{sigma kinds}, we could unify the
signature and carrier translations by translating the kind $\CTy$ to
the sigma kind $\sum_{\code{B} : \CTy} \AThk~(\code{Alg}~\code{T}~\code{B})$. This would have
the benefit that value types would not require any structure to be
defined for them as we could simply translate $\VTy$ to $\VTy$ without
requiring any structure, removing the need for the trivial $\BTop$
structures to be defined or passed around. A downside is that to
support full overloadability we would need to translate these sigma
kinds as well, and that the type system as a whole would be more
complex.
It also may be possible to use Haskell-style typeclasses for algebras
of a relative monad and use a metaprogramming tool like Haskell's
\code{deriving} mechanism to implement our algebra translation.

An avenue for future work would be to extend our algebra translation
to dependently typed CBPV, which would allow us to build the monad and
algebra laws into the types. This would be closely related to the
``weaning translation'' for $\partial$-CBPV
\cite{pedrottabareau17,pedrot-tabareau-2019}. In particular, the
weaning translation is an adaptation of the algebra semantics of CBPV
for an ordinary monad to the setting of intensional type theory, so
our relative monad algebra semantics could possibly be combined to
generalize this to a ``relative weaning translation''.

\subsection{Relative Comonads}

An obvious direction for future work would be to consider relative
\emph{comonads} in CBPV, which would dually be type constructors
$\CTy \to \VTy$ with extra structure, with $\AThk$ being the
prototypical example in the same way that $\BRet$ is the prototypical
relative monad. In the same way that relative monads abstract the
structure of continuations implemented on the stack to allow for
additional effects, relative comonads would abstract the structure
of \emph{closures} to allow for additional \emph{coeffects}.

%% Acknowledgments

\begin{acks}
  The authors would like to thank Tingting Ding, Chengjun Peng, Nathan Varner, Yuxuan Xia, and Eric Zhao for their contributions to the Zydeco implementation.
\end{acks}

%% \begin{acks}                            %% acks environment is optional
%%   %% contents suppressed with 'anonymous'
%% %% Commands \grantsponsor{<sponsorID>}{<name>}{<url>} and
%% %% \grantnum[<url>]{<sponsorID>}{<number>} should be used to
%% %% acknowledge financial support and will be used by metadata
%% %% extraction tools.
%% This material is based upon work supported by the
%% \grantsponsor{GS100000001}{National Science
%% Foundation}{http://dx.doi.org/10.13039/100000001} under Grant
%% No.~\grantnum{GS100000001}{nnnnnnn} and Grant
%% No.~\grantnum{GS100000001}{mmmmmmm}.  Any opinions, findings, and
%% conclusions or recommendations expressed in this material are those
%% of the author and do not necessarily reflect the views of the
%% National Science Foundation.
%% \end{acks}

\newpage
\section{Data-Availability Statement}

An implementation of the Zydeco language that closely resembles the
\cbpv{} calculus is available in a Zenodo repository \cite{zydeco}.
The implementation consists of a parser, a bidirectional type checker,
and a stack-based environment-capturing interpreter. The implementation
is written in Rust and includes all examples of relative monads,
monadic blocks, and algebra translation mentioned in this paper.

% % TBD
% This paper comes with a complete implementation of the Zydeco
% language, which closely resembles the
% \cbpv{} calculus. The Zydeco
% language consists of a parser, a bidirectional type checker, and a
% stack-based environment-capturing interpreter. The implementation is
% written in Rust and is accompanied by functional examples of monadic
% blocks and algebra translation.

%% Bibliography
\bibliography{bib}

\ifextended

\newpage
%% Appendix
\appendix

\section{Monad Implementations}
\label{sec:appendix:monad-implementations}

We demonstrate the implementation of the continuation monad, it
polymorphic version, the state monad, and its continuation-passing
version in \cref{fig:mo_cont_state_impl}.

\begin{figure}
  \begin{align*}
    \code{mkont}~&\defeq~\{~\MTyLam{(B:\CTy)}\code{comatch} \\
    &\qquad\quad\begin{aligned}
      &|~\code{.return}~\Rightarrow~\MTyLam{A}\MTmLam{a} \\
        &\qquad\MTmLam{k}{\MForce{k}~a} \\
      &|~\code{.bind}~\Rightarrow~\MTyLam{A~A'}\MTmLam{m~f} \\
        &\qquad\MTmLam{k}\MForce{m}~\VThunk{~\MTmLam{a}{\MForce{f}~a~k}~}~\}
    \end{aligned}
    \\
    \code{mpolykont}~&\defeq~\{~\code{comatch} \\
    &\qquad\quad\begin{aligned}
      &|~\code{.return}~\Rightarrow~\MTyLam{A}\MTmLam{a} \\
        &\qquad\MTyLam{R}\MTmLam{k}{\MForce{k}~a} \\
      &|~\code{.bind}~\Rightarrow~\MTyLam{A~A'}\MTmLam{m~f} \\
        &\qquad\MTyLam{R}\MTmLam{k}\MForce{m}~\VThunk{~\MTmLam{a}{\MForce{f}~a~k}~}~\}
    \end{aligned}
    \\
    \code{mstate}~&\defeq~\{~\MTyLam{S}\code{comatch} \\
    &\qquad\quad\begin{aligned}
      &|~\code{.return}~\Rightarrow~\MTyLam{A}\MTmLam{a} \\
        &\qquad\MTmLam{s}\MReturn{\VPair{a}{s}} \\
      &|~\code{.bind}~\Rightarrow~\MTyLam{A~A'}\MTmLam{m~f} \\
        &\qquad\MTmLam{s}\MBind{\VPair{a'}{s'}}{\MForce{m}~s}{
          \MForce{f}~a'~s'
        }~\}
    \end{aligned}
    \\
    \code{mstatekont}~&\defeq~\{~\MTyLam{S}\code{comatch} \\
    &\qquad\quad\begin{aligned}
      &|~\code{.return}~\Rightarrow~\MTyLam{A}\MTmLam{a} \\
        &\qquad\MTyLam{R}\MTmLam{k~s}{\MForce{k}~a~s} \\
      &|~\code{.bind}~\Rightarrow~\MTyLam{A~A'}\MTmLam{m~f} \\
        &\qquad\MTyLam{R}\MTmLam{k~s}\MForce{m}~R~\VThunk{~\MTmLam{a~s'}{\MForce{f}~a~R~k~s'}~}~s~\}
    \end{aligned}
    \\
  \end{align*}
  \caption{Implementation of the Continuation and State Monads in \cbpv{}}
  \label{fig:mo_cont_state_impl}
\end{figure}

In \cref{fig:mo_free}, \cref{fig:free_abs_app}, and
\cref{fig:main_free}, we demonstrate an example of the free monad in
\cbpv{}, defining the operations as user interface and handlers as
implementations, and show how to use use them in an example program.

\begin{figure}
  \begin{align*}
    \code{KPrint}~&\defeq~\SLam{R: \CTy}{\AThunk{(\AString \to \code{Kont}~R~\AUnit)}} \\
    \code{KFlip}~&\defeq~\SLam{R: \CTy}{\AThunk{(\AUnit \to \code{Kont}~R~\ABool)}} \\
    \code{PrintFlip}~&\defeq~\SLam{A}{\BForall{R: \CTy}{\code{KPrint}~R \to \code{KFlip}~R \to \code{Kont}~R~A}}
  \end{align*}
  \begin{align*}
    \code{mfree}~&\defeq~\{~\keyword{comatch} \\
      &\quad\quad\begin{aligned}
        &|~\code{.return}~\Rightarrow~\MTyLam{A}{\MTmLam{a}{ \\
          &\qquad\MTyLam{R}{\MTmLam{print~flip}{
            \MTmLam{k}{\MForce{k}~a}~}}}} \\
        &|~\code{.bind}~\Rightarrow~\MTyLam{A~A'}{\MTmLam{m~f}{ \\
          &\qquad\MTyLam{R}{\MTmLam{print~flip}{ \\
            &\qquad\quad\MTmLam{k'}{
              \MForce{m}~R~print~flip \\
            &\qquad\quad\quad\{~\MTmLam{a}{\MForce{f}~a~R~print~flip~k'}~\}}
          }}
        }}~\}
      \end{aligned}
  \end{align*}
  \caption{Example of the Free Monad in \cbpv{}}
  \label{fig:mo_free}
\end{figure}

\begin{figure}
  \begin{align*}
    \code{print}~&\defeq~\{~\MTmLam{s}{\MTyLam{R}{\MTmLam{print~flip}{\MForce{print}~s}}}~\} \\
    \code{flip}~&\defeq~\{~\MTmLam{\_}{\MTyLam{R}{\MTmLam{print~flip}{\MForce{flip}~\VUnit}}}~\} \\
    \code{print\_os}~&\defeq~\{~\MTmLam{s}{\MTmLam{k}{\MForce{\code{write\_line}}~s~\{~\MForce{k}~\VUnit~\}}}~\} \\
    \code{flip\_os}~&\defeq~\{~\MTmLam{\_}{\MTmLam{k}{\MForce{\code{random\_int}}~\VThunk{ \\
      &\qquad\quad\MTmLam{i}{\MBind{j}{\MForce{\code{mod}}~i~2}{\MBind{b}{\MForce{\code{int\_eq}}~j~0}{\MForce{k}~b}}}}}}~\}
  \end{align*}
  \caption{Operations and Handlers for a Free Monad}
  \label{fig:free_abs_app}
\end{figure}

\begin{figure}
  \begin{align*}
    \code{main}~&~\defeq~\{~\MLet{m}{\{
        \\ &\qquad\quad~\MForce{\code{mfree}}~\code{.bind}~\ABool~\AUnit~\VThunk{\MForce{\code{flip}}~\VUnit}~\{~\MTmLam{b}{}
        \\ &\qquad\quad~\MForce{\code{if}}~(\code{PrintFlip}~\AUnit)~b~\VThunk{\MForce{\code{print}}~\code{"+"}}~\VThunk{\MForce{\code{print}}~\code{"-"}}
        \\ &\qquad\quad~\}\}}{\MForce{m}~\BOS~\code{print\_os}~\code{flip\_os}~\VThunk{\MTmLam{\_}{\MForce{\VHalt}}~}
      }
    \}
  \end{align*}
  \caption{Example Program using the Free Monad in \cbpv{}}
  \label{fig:main_free}
\end{figure}

\section{Monad Laws}
\label{sec:appendix:monad-laws}

We provide a counter-example to the right unit monad law for the
defunctionalized continuation monad in
Figure~\ref{fig:exn_kont_counter}.
\begin{figure}
  % \begin{minipage}{0.2\textwidth}
  % \end{minipage}
  % \begin{minipage}{0.48\textwidth}
  %   % def ! rec count_kont
  %   %   (E: VType) (A: VType) (i: Int)
  %   %   (fi: Thunk (Int -> ExnDe E A))
  %   % : ExnDe E A =
  %   %   comatch
  %   %   | .try -> fn E' ke ->
  %   %     ! count_kont E' A i
  %   %     { fn i -> ! fi i .try E' ke }
  %   %   | .kont -> fn A' ka ->
  %   %     do i' <- ! add i 1;
  %   %     ! count_kont E A' i'
  %   %     { fn i -> ! fi i .kont A' ka }
  %   %   | .done ->
  %   %     ! fi i .done
  %   %   end
  %   % end
  %   \inputminted{Haskell}{examples/code/exn_kont_counter.zydeco}
  % \end{minipage}
  \begin{align*}
    &      \code{count\_kont}~\defeq~\{~\MFix{\code{count\_kont}}{ \\
      &\quad\MTyLam{(E: \VTy)~(A: \VTy)}{\MTmLam{(i: \AInt)~(fi: \AThunk{(\AInt \to \code{ExnDe}~E~A)})}{\keyword{comatch} \\
      &\quad\qquad\begin{aligned}
        &|~\code{.try}~\Rightarrow \MTyLam{E'}{\MTmLam{ke}{!~\code{count\_kont}~E'~A~i
          ~\VThunk{~\MTmLam{i}{!~fi~i~\code{.try}~E'~ke}~}}} \\
        &|~\code{.kont}~\Rightarrow \MTyLam{A'}{\MTmLam{ka}{
          ~\keyword{do}~i'~\leftarrow~!~\code{add}~i~1; \\
          &\quad!~\code{count\_kont}~E~A'~i'
          ~\VThunk{~\MTmLam{i}{!~fi~i~\code{.kont}~A'~ka}~}
        }} \\
        &|~\code{.done}~\Rightarrow~\MForce{fi~i~\code{.done}}~\}
      \end{aligned}
    }}}
  \end{align*}
  % \begin{minipage}{0.4\textwidth}
  %   % alias ExnCounter = ExnDe Unit Int end

  %   % def ! bench
  %   %   (impl: Thunk (Thunk (ExnCounter) -> ExnCounter))
  %   % : Ret Int =
  %   %   do i? <- ! impl {
  %   %     ! count_kont Unit Int 0
  %   %     { ! mexnde Unit .return Int }
  %   %   } .done;
  %   %   match i?
  %   %   | +Left(e) -> ret -1
  %   %   | +Right(i) -> ret i
  %   %   end
  %   % end
  %   \inputminted{Haskell}{examples/code/exn_kont_counter_bench.zydeco}
  % \end{minipage}
  \begin{align*}
    \code{ExnCounter}~&\defeq~\code{ExnDe}~\AUnit~\AInt \\
    \code{bench}~&\defeq~\{~\MTmLam{(impl: \AThunk{(\AThunk{\code{ExnCounter}} \to \code{ExnCounter})})}{ \\
      &\qquad\quad\begin{aligned}
        &\keyword{do}~i?~\leftarrow~!~\code{impl}~\{ \\
        &\quad!~\code{count\_kont}~\AUnit~\AInt~0 \\
        &\quad\{~!~\code{mexnde}~\AUnit~\code{.return}~\AInt~\} \\
        &\}~\code{.done}; \\
        &\keyword{match}~i?
        ~|\,\,\,\code{Err}(e) \Rightarrow \keyword{ret}~(-1)
        ~|\,\,\,\code{Ok}(i) \Rightarrow \keyword{ret}~i~\}
      \end{aligned}
    }
  \end{align*}
  \caption{Counter Example Counter-Example}
  \label{fig:exn_kont_counter}
\end{figure}

We verify the monad laws on the exception monads satisfying the
canonicity condition.  
The right unit law is satisfied without the extra conditions:
\begin{align*}
         &\MForce{m}~\code{.bind}~A~A'~\VThunk{\MForce{m}~\code{.return}~A~a}~f \\
  \equiv~&\MForce{m~\code{.return}}~A~a~\code{.kont}~A'~f \\
  \equiv~&\MForce{f}~a
\end{align*}
And the left unit law can be verified if we apply the canonicity condition on $\MForce{m}$.
\begin{align*}
       &\MForce{m~\code{.bind}}~A~A~m~\VThunk{\MForce{m~\code{.return}}~A} \\
\equiv~&\MForce{m}~\code{.kont}~A~\VThunk{\MForce{m~\code{.return}}~A} \\
\equiv~&\MBind{x}{\MForce{m}}{} \\
       &\keyword{match}~x \\
       &|\,\,\,\code{Err}(e) \to \MForce{\code{fail}}~e \\
       &|\,\,\,\code{Ok}(a) \to \MForce{m~\code{.return}}~A~a \\
       &\code{.kont}~A~\VThunk{\MForce{m~\code{.return}}~A} \\
\equiv~&\MBind{x}{\MForce{m}}{} \\
       &\keyword{match}~x \\
       &|\,\,\,\code{Err}(e) \to \MForce{\code{fail}}~e~\code{.kont}~A~\VThunk{\MForce{m~\code{.return}}~A} \\
       &|\,\,\,\code{Ok}(a) \to \MForce{m~\code{.return}}~A~a~\code{.kont}~A~\VThunk{\MForce{m~\code{.return}}~A} \\
\equiv~&\MBind{x}{\MForce{m}}{} \\
       &\keyword{match}~x \\
       &|\,\,\,\code{Err}(e) \to \MForce{\code{fail}}~e \\
       &|\,\,\,\code{Ok}(a) \to \MForce{m~\code{.return}}~A~a \\
\equiv~&\MForce{m}
\end{align*}
As for the associativity law, we observe
\begin{align*}
\MForce{m~\code{.bind}}~A'~A''~\VThunk{\MForce{m~\code{.bind}}~A~A'~m~f}~g~&\equiv~\MForce{m}~\code{.kont}~A'~f~\code{.kont}~A''~g \\
\MForce{m~\code{.bind}}~A~A'~m~\VThunk{\MTmLam{x}{\MForce{m~\code{.bind}}~A'~A''~\VThunk{\MForce{f}~x}~g}}~&\equiv~\MForce{m}~\code{.kont}~A''~\VThunk{\MTmLam{x}{\MForce{f}~x~\code{.kont}~A''~g}}
\end{align*}
By applying the canonicity condition on $\MForce{m}$ and using the
same reasoning as in the left unit law, we have for the left side
\begin{align*}
       &\MForce{m}~\code{.kont}~A'~f~\code{.kont}~A''~g \\
\equiv~&\MBind{x}{\MForce{m}}{} \\
       &\keyword{match}~x \\
       &|\,\,\,\code{Err}(e) \to \MForce{\code{fail}}~e \\
       &|\,\,\,\code{Ok}(a) \to \MForce{m~\code{.return}}~A~a~\code{.kont}~A'~f~\code{.kont}~A''~g \\
\equiv~&\MBind{x}{\MForce{m}}{} \\
       &\keyword{match}~x \\
       &|\,\,\,\code{Err}(e) \to \MForce{\code{fail}}~e \\
       &|\,\,\,\code{Ok}(a) \to \MForce{f}~a~\code{.kont}~A''~g \\
\end{align*}
and the right side
\begin{align*}
       &\MForce{m}~\code{.kont}~A''~\VThunk{\MTmLam{x}{\MForce{f}~x~\code{.kont}~A''~g}} \\
\equiv~&\MBind{x}{\MForce{m}}{} \\
       &\keyword{match}~x \\
       &|\,\,\,\code{Err}(e) \to \MForce{\code{fail}}~e \\
       &|\,\,\,\code{Ok}(a) \to \MForce{m~\code{.return}}~A~a~\code{.kont}~A''~\VThunk{\MTmLam{x}{\MForce{f}~x~\code{.kont}~A''~g}} \\
\equiv~&\MBind{x}{\MForce{m}}{} \\
       &\keyword{match}~x \\
       &|\,\,\,\code{Err}(e) \to \MForce{\code{fail}}~e \\
       &|\,\,\,\code{Ok}(a) \to \MForce{f}~a~\code{.kont}~A''~g \\
\end{align*}

Finally, the linearity law trivially holds because $M$ is syntactically restricted
to be used for only once in the computation.
The only surrounding computation of $M$ is a do-binding, which preserves linearity.
Therefore, we conclude that all four monad laws are satisfied with the canonicity property.

\section{Algebra Translation}
\label{sec:appendix:algebra-translatiosn}

We give the full algebra translation composed of seven translations.
\begin{enumerate}
  \item Signature Translation: see
        \cref{fig:appendix:signature-translation}.
  \item Type Environment Signature Translation: see
        \cref{fig:appendix:signature-translation}.
  \item Carrier Translation: see
        \cref{fig:appendix:carrier-translation}.
  \item Value Environment Translation: see
        \cref{fig:appendix:carrier-translation}.
  \item Structure Translation: see
        \cref{fig:appendix:structure-translation}.
  \item Term Translations: see
        \cref{fig:appendix:term-translation-value} and
        \cref{fig:appendix:term-translation-computation}.
  \item Monadic Block Translation: see
        \cref{fig:appendix:monadic-block-translation}.
\end{enumerate}

\begin{figure}
\begin{align*}
  \textrm{Sig}~&:~\forall~\Implicit{K}~\Implicit{\Delta} \to (\code{T}: \VTy \to \CTy, \Delta \vdash K) \to (T: \VTy \to \CTy, \Delta \vdash \CTy) \\
  \SigP{\VTy}{A}~&:=~{\BTop} \\
  \SigP{\CTy}{B}~&:=~{\BAlgebra{T}{B}} \\
  \SigP{K_0 \to K}{S}~&:=~{\BForall{X:K_0}{\AThunk{(\SigP{K_0}{X})} \to \SigP{K}{S~X}}} \\
  \\
  \textrm{Sig}~&:~(\Delta : \textrm{TEnv}) \to \Delta \vdash \textrm{VEnv} \\
  \SigD{\,\cdot\,}~&:=~\cdot \\
  \SigD{\Delta, X:K}~&:=~\SigD{\Delta}, \code{str}_{\code{X}}: \AThunk{(\SigP{K}{X})}
\end{align*}
\caption{Signature Translation and Type Environment Signature Translation}
\label{fig:appendix:signature-translation}
\end{figure}

\begin{figure}
\begin{align*}
  \floor{\,\cdot\,}~&:~\forall~\Implicit{K}~\Implicit{\Delta} \to (\Delta \vdash K) \to (T: \VTy \to \CTy, \Delta \vdash K) \\
  \floor{\AThk}~&:=~\AThk \\
  \floor{\AUnit}~&:=~\AUnit \\
  \floor{(\times)}~&:=~(\times) \\
  \floor{(+)}~&:=~(+) \\
  \floor{\AExists{X:K}{A}}~&:=~\AExists{X:K}{\AThunk{(\SigP{K}{X})} \times \floor{A}} \\
  \floor{\BRet}~&:=~T \\
  \floor{(\to)}~&:=~(\to) \\
  \floor{(\&_D)}~&:=~(\&_D) \\
  \floor{\BForall{X:K}{B}}~&:=~\BForall{X:K}{\AThunk{(\SigP{K}{X})} \to \floor{B}} \\
  \floor{X}~&:=~X \\
  \floor{\SLam{X:K}{S}}~&:=~\SLam{X:K}{\floor{S}} \\
  \floor{S~S_0}~&:=~\floor{S}~\floor{S_0} \\
  \floor{\BNu{Y:K}{B}}~&:=~\BNu{Y:K}{\floor{B}} \\
  \\
  \floor{\,\cdot\,}~&:~VEnv \to VEnv \\
  \floor{\,\cdot\,}~&:=~\cdot \\
  \floor{\Gamma, x: A}~&:=~\floor{\Gamma}, x: \floor{A}
\end{align*}
\caption{Carrier Translation and Value Environment Translation}
\label{fig:appendix:carrier-translation}
\end{figure}

\begin{figure}
\begin{align*}
  Str~&:~\forall~\Implicit{K}~\Implicit{\Delta} \to (\Delta \vdash S: K) \\
      &\,\,\,\,\to (T: \VTy \to \CTy, \Delta;~mo: \AThunk{(\BMonad{T})}, \SigD{\Delta} \vdash \SigP{K}{\floor{S}}) \\
  \StrP{\AThk}~&:=~\MTyLam{X}{\MTmLam{\_}{\MTop}} \\
  \StrP{\AUnit}~&:=~\MTop \\
  \StrP{(\times)}~&:=~\MTyLam{X}\MTmLam{\_}\MTyLam{Y}\MTmLam{\_}\MTop \\
  \StrP{(+_C)}~&:=~\multi{\MTyLam{X}\MTmLam{\_}}_{c \in C}~\MTop \\
  \StrP{\AExists{X:K}{A}}~&:=~\MTop \\
  \StrP{\BRet}~&:=~\MTyLam{X}{\MTmLam{\_}{\MTyLam{Z}{\MTmLam{tz~(f: \AThunk{(Z \to T~X)})}{
    \\ &\qquad\,\,\,\,\,\,\MForce{\code{m}}~\code{.bind}~Z~X~tz~f
  }}}} \\
  \StrP{(\to)}~&:=~\MTyLam{X}{\MTmLam{\_}{\MTyLam{Y}{\MTmLam{alg_Y}{
    \MTyLam{Z}{\MTmLam{tz~(f: \AThunk{(Z \to \floor{X} \to \floor{Y})})}{
      \\ &\qquad\quad\MTmLam{(x: \floor{X})}{\MForce{alg_Y}~Z~tz~\VThunk{\MTmLam{z}{\MForce{f}~z~x}}}
    }}
  }}}} \\
  \StrP{(\&_D)}~&:=~\multi{\MTyLam{X_d}\MTmLam{alg_d}}_{d \in D}
    \MTyLam{Z}\MTmLam{tz~(f: \AThunk{(Z \to \BWith{.d:\floor{X_d}}{d \in D})})}
      \\ &\qquad\quad\MComatch
        {.d}
          {\MForce{alg_d}~Z~tz~\VThunk{\MTmLam{z}{\MForce{f}~z~.d}}}
        {d \in D}
  \\
  \StrP{\BForall{X:K}{B}}~&:=~\MTyLam{Z}{\MTmLam{tz~(f: \AThunk{(Z \to \BForall{X:K}{{\AThunk(\SigP{K}{X})} \to \floor{B}})})}{
    \\ &\qquad\quad\MTyLam{X:K}{\MTmLam{(\code{str}_{\code{X}}: \AThunk{(\SigP{K}{X})})}{ \StrP{B}~Z~tz~\VThunk{\MTmLam{z}{ \MForce{f}~z~X~\code{str}_{\code{X}} }} }}
  }} \\
  \StrP{\BNu{Y:K}{B}}~&:=~\MFix{(str : \AThunk{(\BAlgebra{T}{(\BNu{Y:K}{\floor{B}})})})}{
    \\ &\qquad\quad\MTyLam{Z}{\MTmLam{tz~(f:\AThunk{Z \to \BNu{Y:K}{\floor{B}}})}{
      \\ &\qquad\quad\quad\MRoll{
        \subst
          {\subst{\StrP{B}}{(\BNu{Y:K}{\floor{B}})}{Y}}
          {str}{str_Y}
          ~Z~tz~\VThunk{\MTmLam{z}{\MUnroll{\MForce{f}~z~}}}
      }
    }
  }} \\
  \StrP{X}~&:=~\MForce{\code{str}_{\code{X}}} \\
  \StrP{\SLam{X:K}{S}}~&:=~\MTyLam{X}{\MTmLam{\code{str}_{\code{X}}}{\StrP{S}}} \\
  \StrP{S~S_0}~&:=~\StrP{S}~\floor{S_0}~\VThunk{\StrP{S_0}}
\end{align*}
\caption{Structure Translation}
\label{fig:appendix:structure-translation}
\end{figure}

\begin{figure}
\begin{align*}
  \floor{\,\cdot\,}~&:~\forall~\Implicit{A}~\Implicit{\Delta}~\Implicit{\Gamma} \\
                    &\,\,\,\,\to (~\Delta; \Gamma \vdash A~) \\
                    &\,\,\,\,\to (~T: \VTy \to \CTy, \Delta \\
                    &\qquad\,\,;~mo: \AThunk{(\BMonad{T})}, \SigD{\Delta}, \floor{\Gamma} \vdash \floor{A}~) \\
  \floor{x}~&:=~x \\
  \floor{\VThunk{M}}~&:=~\VThunk{\floor{M}} \\
  \floor{\VUnit}~&:=~\VUnit \\
  \floor{\VPair{V_0}{V_1}}~&:=~\VPair{\floor{V_0}}{\floor{V_1}} \\
  \floor{\VInj{c}{V}}~&:=~\VInj{c}{\floor{V}} \\
  \floor{\VPack{S}{V}}~&:=~\VPack{\floor{S}}{\VPair{\VThunk{\StrP{S}}}{\floor{V}}}
\end{align*}
\caption{Term Translation (Values)}
\label{fig:appendix:term-translation-value}
\end{figure}

\begin{figure}
\begin{align*}
  \floor{\,\cdot\,}~&:~\forall~\Implicit{B}~\Implicit{\Delta}~\Implicit{\Gamma} \\
                    &\,\,\,\,\to (~\Delta; \Gamma \vdash B~) \\
                    &\,\,\,\,\to (~T: \VTy \to \CTy, \Delta \\
                    &\qquad\,\,;~mo: \AThunk{(\BMonad{T})}, \SigD{\Delta}, \floor{\Gamma} \vdash \floor{B}~) \\
  \floor{\MForce{V}}~&:=~\MForce{\floor{V}} \\
  \floor{\MLetPair{x_0}{x_1}{V}{M}}~&:=~\MLetPair{x_0}{x_1}{\floor{V}}{\floor{M}} \\
  \floor{\MMatch{V}{\Inj{c}(x_c)}{M_c}{c \in C}}~&:=~\MMatch{\floor{V}}{\Inj{c}(x_c)}{\floor{M_c}}{c \in C} \\
  \floor{\MLetPack{X}{x}{V}{M}}~&:=~\MLetPack{X}{p}{\floor{V}}{\MLetPair{\code{str}_{\code{X}}}{x}{p}{\floor{M}}} \\
  \floor{\MReturn{V}}~&:=~\,\MForce{\code{m}}~\code{.return}~\floor{A}~\floor{V}
    \\ &\qquad\text{where}~\Delta;\Gamma \vdash V:A \\
  \floor{\MBind{x}{M_0}{M}}~&:=~\StrP{B}~\floor{A}~\VThunk{\floor{M_0}}~\VThunk{\MTmLam{x}{\floor{M}}}
    \\ &\qquad\text{where}~\Delta;\Gamma \vdash M_0:\BReturn{A}~\text{and}~\Delta;\Gamma \vdash M:B \\
  \floor{\MTmLam{x:A}{M}}~&:=~\MTmLam{x:\floor{A}}{\floor{M}} \\
  \floor{M~V}~&:=~\floor{M}~\floor{V} \\
  \floor{\MTop}~&:=~\MTop \\
  \floor{\MComatch{.d}{M_d}{d \in D}}~&:=~\MComatch{.d}{\floor{M_d}}{d \in D} \\
  \floor{M~.i}~&:=~\floor{M}~.i \\
  \floor{\MTyLam{X:K}{M}}~&:=~\MTyLam{X:K}{\MTmLam{\code{str}_{\code{X}}}\floor{M}} \\
  \floor{M~S}~&:=~\floor{M}~\floor{S}~\VThunk{\StrP{S}} \\
  \floor{\MRoll{M}}~&:=~\MRoll{\floor{M}} \\
  \floor{\MUnroll{M}}~&:=~\MUnroll{\floor{M}} \\
  \floor{\MFix{x}{M}}~&:=~\MFix{x}{\floor{M}}
\end{align*}
\caption{Term Translation (Computations)}
\label{fig:appendix:term-translation-computation}
\end{figure}

\begin{figure}
\begin{align*}
  \brac{\,\cdot\,}~&:~\forall~\Implicit{A}~\Implicit{\Delta}~\Implicit{\Gamma}
    \to (\Delta; \Gamma \vdash A)
    \to (\Delta; \Gamma \vdash A) \\
  \brac{x}~&:=~x \\
  \brac{\VThunk{M}}~&:=~\VThunk{\brac{M}} \\
  \brac{\VUnit}~&:=~\VUnit \\
  \brac{\VPair{V_0}{V_1}}~&:=~\VPair{\brac{V_0}}{\brac{V_1}} \\
  \brac{\VInj{c}{V}}~&:=~\VInj{c}{\brac{V}} \\
  \brac{\VPack{S}{V}}~&:=~\VPack{S}{\brac{V}} \\
  \\
  \brac{\,\cdot\,}~&:~\forall~\Implicit{B}~\Implicit{\Delta}~\Implicit{\Gamma}
    \to (\Delta; \Gamma \vdash B)
    \to (\Delta; \Gamma \vdash B) \\
  \brac{\MMonadic{M}}~&:=~\MTyLam{(T: \VTy \to \CTy)}{ \\
    &~\quad\quad\,\, \MTmLam{(mo: \AThunk{(\BMonad{T})})}{\floor{\brac{M}}}} \\
  \brac{\MForce{V}}~&:=~\MForce{\brac{V}} \\
  \brac{\MLetPair{x_0}{x_1}{V}{M}}~&:=~\MLetPair{x_0}{x_1}{\brac{V}}{\brac{M}} \\
  \brac{\MMatch{V}{\Inj{c}(x_c)}{M_c}{c \in C}}~&:=~\MMatch{\brac{V}}{\Inj{c}(x_c)}{\brac{M_c}}{c \in C} \\
  \brac{\MLetPack{X}{x}{V}{M}}~&:=~\MLetPack{X}{x}{\brac{V}}{\brac{M}} \\
  \brac{\MReturn{V}}~&:=~\MReturn{\brac{V}} \\
  \brac{\MBind{x}{M_0}{M}}~&:=~\MBind{x}{\brac{M_0}}{\brac{M}} \\
  \brac{\MTmLam{x}{M}}~&:=~\MTmLam{x}{\brac{M}} \\
  \brac{M~V}~&:=~\brac{M}~\brac{V} \\
  \brac{\MTop}~&:=~\MTop \\
  \brac{\MComatch{.d}{M_d}{d \in D}}~&:=~\MComatch{.d}{\brac{M_d}}{d \in D} \\
  \brac{M~.i}~&:=~\brac{M}~.i \\
  \brac{\MTyLam{X:K}{M}}~&:=~\MTyLam{X}{\brac{M}} \\
  \brac{M~S}~&:=~\brac{M}~S \\
  \brac{\MRoll{M}}~&:=~\MRoll{\brac{M}} \\
  \brac{\MUnroll{M}}~&:=~\MUnroll{\brac{M}} \\
  \brac{\MFix{x}{M}}~&:=~\MFix{x}{\brac{M}}
\end{align*}
\caption{Monadic Block Translation}
\label{fig:appendix:monadic-block-translation}
\end{figure}

\newpage~
\newpage~
\newpage~
\newpage~

\section{Algebras Generated by The Algebra Translation are Lawful}

In this section we show that the algebra structures generated by the
algebra translation obey the algebra laws. The translation in concern
is the term translation, defined in
\cref{fig:appendix:term-translation-value} and
\cref{fig:appendix:term-translation-computation}. In order to prove
all algebras we construct are lawful, we need to generalize the
inductive hypothesis and prove all structures satisfy some ``laws'':
\begin{enumerate}
  \item For computation types, the structure is an algebra and should obey the algebra laws;
  \item For value types, the structure is $\BTop$ and so no laws need to be verified.
  \item For type constructors $K \to K'$, we need to verify that when
    instantiated with a lawful structure for $K$, we produce a lawful
    structure for $K'$.
\end{enumerate}

We first talk about value types and type constructors.  
The value types trivially obey the laws since there is no law to obey.  
For type constructors, since the lawful structure is passed in through
a function parameter, we only care about the cases when a computation
type is returned.  
Therefore, we only need to verify that the laws hold for computation types.

\begin{lemma}
  The algebra generated by the algebra translation is lawful under
  CBPV $\beta\eta$ equality. Namely, the following equational
  principles hold for all type environment $\Delta$, monad type
  constructor $\Delta \vdash T: \VTy \to \CTy$, and computation
  type $\Delta \vdash B: \CTy$:
  \begin{enumerate}
    \item \textbf{left unit law:} \[
      \StrP{B}~A~\VThunk{\MForce{m~\code{.return}}~A~a}~f \equiv \MForce{f}~a
    \] for any $a: \AMV$ and $f: \AThunk{(\AMV \to \SigP{\CTy}{\BMV})}$,
    \item \textbf{associativity law:} \[
      \StrP{B}~A'~\VThunk{\MForce{m~\code{.bind}}~A~A'~t~f}~g \equiv \StrP{B}~A~t~\VThunk{\MTmLam{x}{\StrP{B}~A'~\VThunk{\MForce{f}~x}~g}}
    \] for any $f: \AThunk{(\AMV \to T~\AMV')}$, $g: \AThunk{(\AMV' \to \SigP{\CTy}{\BMV})}$, and $t: \AThunk{(T~\AMV)}$,
    \item \textbf{linearity law:} \[
      \MBind{t}{\MForce{tt}}{\StrP{B}~A~t} \equiv \StrP{B}~A~\VThunk{\MBind{t}{\MForce{tt}}{\MForce{t}}}
    \] for any $tt: \AThunk{(\BReturn{(\AThunk{(T~\AMV)})})}$.
  \end{enumerate}
\end{lemma}

\begin{proof}
By induction on the cases of structure translation.

We start from the left unit law:
\begin{enumerate}
  \item \textbf{Case} $\StrP{\BReturn{A_0}}$ for $A_0: \VTy$: by applying the left unit law of monad $T$
    \begin{align*}
      &\StrP{\BReturn{A_0}}~A~\VThunk{\MForce{m~\code{.return}}~A~a}~f \\
      \equiv~&\MForce{m~\code{.bind}}~A~A_0~\VThunk{\MForce{m~\code{.return}}~A~a}~f \\
      \equiv~&\MForce{f}~a
    \end{align*}
  \item \textbf{Case} $\StrP{A_0 \to B}$ for $A_0: \VTy$ and $B: \CTy$: by induction hypothesis on $B$
    \begin{align*}
      &\StrP{A_0 \to B}~A~\VThunk{\MForce{m~\code{.return}}~A~a}~f \\
      \equiv~&\MTmLam{a_0}{\StrP{A_0 \to B}~A~\VThunk{\MForce{m~\code{.return}}~A~a}~f~a_0} \\
      \equiv~&\MTmLam{a_0}{\StrP{B}~A~\VThunk{\MForce{m~\code{.return}}~A~a}~\VThunk{\MTmLam{x}{\MForce{f}~x~a_0}}} \\
      \equiv~&\MTmLam{a_0}{(\MTmLam{x}{\MForce{f}~x~a_0})~a} \\
      \equiv~&\MTmLam{a_0}{\MForce{f}~a~a_0} \\
      \equiv~&\MForce{f}~a
    \end{align*}
  \item \textbf{Case} $\StrP{\BWith{B_d}{d \in D}}$ for $\multi{B_d: \CTy}_{d \in D}$: by induction hypothesis on $\multi{B_d}_{d \in D}$
    \begin{align*}
      &\StrP{\BWith{B_d}{d \in D}}~A~\VThunk{\MForce{m~\code{.return}}~A~a}~f \\
      \equiv~&\MComatch{.d}
        {\StrP{B_d}~A~\VThunk{\MForce{m~\code{.return}}~A~a}~\VThunk{\MTmLam{x}{\MForce{f}~x~.d}}}
      {d \in D} \\
      \equiv~&\MComatch{.d}{\MForce{f}~a~.d}{d \in D} \\
      \equiv~&\MForce{f}~a
    \end{align*}
  \item \textbf{Case} $\StrP{\BForall{X:K}{B}}$ for all $K$ and $B: \CTy$: by induction hypothesis on $B$
    \begin{align*}
      &\StrP{\BForall{X:K}{B}}~A~\VThunk{\MForce{m~\code{.return}}~A~a}~f \\
      \equiv~&\MTyLam{X}{\MTmLam{\code{str}_{\code{X}}}{\StrP{B}~A~\VThunk{\MForce{m~\code{.return}}~A~a}~\VThunk{\MTmLam{x}{\MForce{f}~x~X~\code{str}_{\code{X}}}}}} \\
      \equiv~&\MTyLam{X}{\MTmLam{\code{str}_{\code{X}}}{\MForce{f}~a~X~\code{str}_{\code{X}}}} \\
      \equiv~&\MForce{f}~a
    \end{align*}
  \item \textbf{Case} $\StrP{\BNu{Y:\CTy}{B}}$ for $B: \CTy$: \newline
    let $str_\nu$ be $\VThunk{\StrP{\BNu{Y:\CTy}{B}}}$ \newline
    by induction hypothesis on $B$
    \begin{align*}
      &\StrP{\BNu{Y:\CTy}{B}}~A~\VThunk{\MForce{m~\code{.return}}~A~a}~f \\
      \equiv~&\MRoll{
        \subst
          {\subst{\StrP{B}}{B_\nu}{Y}}
          {str_\nu}{str_Y}
          ~A~\VThunk{\MForce{m~\code{.return}}~A~a}~\VThunk{\MTmLam{x}{\MUnroll{\MForce{f}~x~}}}
      } \\
      \equiv~&\MRoll{(\MTmLam{x}{\MUnroll{\MForce{f}~x~}})~a} \\
      \equiv~&\MRoll{\MUnroll{\MForce{f}~a~}} \\
      \equiv~&\MForce{f}~a
    \end{align*}
  \item \textbf{Case} $\StrP{X}$ for $X: \CTy$: since the structure is
    from the value environment, it's lawful by the induction hypothesis.
  \item \textbf{Case} $\StrP{(\SLam{X:K}{B})~S}$ for $S: K$ and $B: \CTy$: by induction hypothesis on $S$ and $B$
    \begin{align*}
      &\StrP{(\SLam{X:K}{B})~S}~A~\VThunk{\MForce{m~\code{.return}}~A~a}~f \\
      \equiv~&(\MTyLam{X}{\MTmLam{\code{str}_{\code{X}}}{\StrP{B}}})~\floor{S}~\VThunk{\StrP{S}}~A~\VThunk{\MForce{m~\code{.return}}~A~a}~f \\
      \equiv~&\subst{\subst{\StrP{B}}{\floor{S}}{X}}{\VThunk{\StrP{S}}}{\code{str}_{\code{X}}}~A~\VThunk{\MForce{m~\code{.return}}~A~a}~f \\
      \equiv~&\MForce{f}~a
    \end{align*}
\end{enumerate}

Then we move on to the associativity law:
\begin{enumerate}
  \item \textbf{Case} $\StrP{\BReturn{A_0}}$ for all $A_0: \VTy$: by applying the associativity law of monad $T$
    \begin{align*}
      &\StrP{\BReturn{A_0}}~A'~\VThunk{\MForce{m~\code{.bind}}~A~A'~m~f}~g \\
      \equiv~&\MForce{m~\code{.bind}}~A'~A_0~\VThunk{\MForce{m~\code{.bind}}~A~A'~m~f}~g \\
      \equiv~&\MForce{m~\code{.bind}}~A~A_0~m~\VThunk{\MTmLam{x}{\MForce{m~\code{.bind}}~A'~A_0~\VThunk{\MForce{f}~x}~g}}\\
      \equiv~&\MForce{m~\code{.bind}}~A~A_0~m~\VThunk{\MTmLam{x}{\StrP{\BReturn{A_0}}~A'~\VThunk{\MForce{f}~x}~g}}\\
      \equiv~&\StrP{\BReturn{A_0}}~A~m~\VThunk{\MTmLam{x}{\StrP{\BReturn{A_0}}~A'~\VThunk{\MForce{f}~x}~g}}
    \end{align*}
  \item \textbf{Case} $\StrP{A_0 \to B}$ for $A_0: \VTy$ and $B: \CTy$: by induction hypothesis on $B$
    \begin{align*}
      &\StrP{A_0 \to B}~A'~\VThunk{\MForce{m~\code{.bind}}~A~A'~m~f}~g \\
      \equiv~&\MTmLam{z}{\StrP{B}~A'~\VThunk{\MForce{m~\code{.bind}}~A~A'~m~f}~\VThunk{\MTmLam{y}{\MForce{g}~y~z}}}\\
      \equiv~&\MTmLam{z}{\StrP{B}~A~m~\VThunk{\MTmLam{x}{\StrP{B}~A'~\VThunk{\MForce{f}~x}~\VThunk{\MTmLam{y}{\MForce{g}~y~z}}}}}\\
      \equiv~&\MTmLam{z}{\StrP{B}~A~m~\VThunk{\MTmLam{x}{(\MTmLam{w}{\StrP{B}~A'~\VThunk{\MForce{f}~x}~\VThunk{\MTmLam{y}{\MForce{g}~y~w}}})~z}}}\\
      \equiv~&\MTmLam{z}{\StrP{B}~A~m~\VThunk{\MTmLam{x}{\StrP{A_0 \to B}~A'~\VThunk{\MForce{f}~x}~g~z}}}\\
      \equiv~&\MTmLam{z}{\StrP{B}~A~m~\VThunk{\MTmLam{w}{\StrP{A_0 \to B}~A'~\VThunk{\MForce{f}~w}~g~z}}}\\
      \equiv~&\MTmLam{z}{\StrP{B}~A~m~\VThunk{\MTmLam{w}{(\MTmLam{x}{\StrP{A_0 \to B}~A'~\VThunk{\MForce{f}~x}~g})~w~z}}}\\
      \equiv~&\StrP{A_0 \to B}~A~m~\VThunk{\MTmLam{x}{\StrP{A_0 \to B}~A'~\VThunk{\MForce{f}~x}~g}}
    \end{align*}
  \item \textbf{Case} $\StrP{\BWith{B_d}{d \in D}}$ for $\multi{B_d: \CTy}_{d \in D}$: by induction hypothesis on $\multi{B_d}_{d \in D}$
    \begin{align*}
      &\StrP{\BWith{B_d}{d \in D}}~A'~\VThunk{\MForce{m~\code{.bind}}~A~A'~m~f}~g \\
      \equiv~&\MComatch{.d}
        {\StrP{B_d}~A'~\VThunk{\MForce{m~\code{.bind}}~A~A'~m~f}~\VThunk{\MTmLam{x}{\MForce{g}~x~.d}}}
      {d \in D} \\
      \equiv~&\MComatch{.d}
        {\StrP{B_d}~A~m~\VThunk{\MTmLam{x}{
          \StrP{B_d}~A'~\VThunk{\MForce{f}~x}~\VThunk{\MTmLam{x}{\MForce{g}~x~.d}}
        }}}
      {d \in D} \\
      \equiv~&\MComatch{.d}
        {\StrP{B_d}~A~m~\VThunk{\MTmLam{x}{\StrP{\BWith{B_d}{d \in D}}~A'~\VThunk{\MForce{f}~x}~g~.d}}}
      {d \in D} \\
      \equiv~&\StrP{\BWith{B_d}{d \in D}}~A~m~\VThunk{\MTmLam{x}{\StrP{\BWith{B_d}{d \in D}}~A'~\VThunk{\MForce{f}~x}~g}}
    \end{align*}
  \item \textbf{Case} $\StrP{\BForall{X:K}{B}}$ for all $K$ and $B: \CTy$: by induction hypothesis on $B$
    \begin{align*}
      &\StrP{\BForall{X:K}{B}}~A'~\VThunk{\MForce{m~\code{.bind}}~A~A'~m~f}~g \\
      \equiv~&\MTyLam{X}{\MTmLam{\code{str}_{\code{X}}}{\StrP{B}~A'~\VThunk{\MForce{m~\code{.bind}}~A~A'~m~f}~\VThunk{\MTmLam{z}{\MForce{g}~z~X~\code{str}_{\code{X}}}}}} \\
      \equiv~&\MTyLam{X}{\MTmLam{\code{str}_{\code{X}}}{\StrP{B}~A~m~\VThunk{\MTmLam{x}{
        \StrP{B}~A'~\VThunk{\MForce{f}~x}~\VThunk{\MTmLam{z}{\MForce{g}~z~X~\code{str}_{\code{X}}}}
      }}}} \\
      \equiv~&\MTyLam{X}{\MTmLam{\code{str}_{\code{X}}}{\StrP{B}~A~m~\VThunk{\MTmLam{x}{\StrP{\BForall{X:K}{B}}~A'~\VThunk{\MForce{f}~x}~g~X~\code{str}_{\code{X}}}}}} \\
      \equiv~&\StrP{\BForall{X:K}{B}}~A~m~\VThunk{\MTmLam{x}{\StrP{\BForall{X:K}{B}}~A'~\VThunk{\MForce{f}~x}~g}}
    \end{align*}
  \item \textbf{Case} $\StrP{\BNu{Y:\CTy}{B}}$ for $B: \CTy$: \newline
    let $B_\nu$ be $\BNu{Y:\CTy}{\floor{B}}$ \newline
    and $str_\nu$ be $\VThunk{\StrP{\BNu{Y:\CTy}{B}}}$ \newline
    by induction hypothesis on $B$
    \begin{align*}
      &\StrP{\BNu{Y:\CTy}{B}}~A'~\VThunk{\MForce{m~\code{.bind}}~A~A'~m~f}~g \\
      \equiv~&\MRoll{
        \subst
          {\subst{\StrP{B}}{B_\nu}{Y}}
          {str_\nu}{str_Y}
          ~A'~\VThunk{\MForce{m~\code{.bind}}~A~A'~m~f}~\VThunk{\MTmLam{x}{\MUnroll{\MForce{g}~x}}}
      } \\
      \equiv~&\MRoll{
        \subst
          {\subst{\StrP{B}}{B_\nu}{Y}}
          {str_\nu}{str_Y}
          ~A~m~\VThunk{\MTmLam{x}{
            \\ &\quad % \StrP{\BNu{Y:\CTy}{B}}~A'~\VThunk{\MForce{f}~x}~g
              \subst
                {\subst{\StrP{B}}{B_\nu}{Y}}
                {str_\nu}{str_Y}
                ~A'~\VThunk{\MForce{f}~x}~\VThunk{\MTmLam{y}{\MUnroll{\MForce{g}~y}}}
          }}
      } \\
      \equiv~&\MRoll{
        \subst
          {\subst{\StrP{B}}{B_\nu}{Y}}
          {str_\nu}{str_Y}
          ~A~m~\VThunk{\MTmLam{x}{\MUnroll{
            \\ &\quad % \StrP{\BNu{Y:\CTy}{B}}~A'~\VThunk{\MForce{f}~x}~g
            \MRoll{
              \subst
                {\subst{\StrP{B}}{B_\nu}{Y}}
                {str_\nu}{str_Y}
                ~A'~\VThunk{\MForce{f}~x}~\VThunk{\MTmLam{y}{\MUnroll{\MForce{g}~y}}}
            }
          }}}
      } \\
      \equiv~&\MRoll{
        \subst
          {\subst{\StrP{B}}{B_\nu}{Y}}
          {str_\nu}{str_Y}
          ~A~m~\VThunk{\MTmLam{x}{\MUnroll{\StrP{\BNu{Y:\CTy}{B}}~A'~\VThunk{\MForce{f}~x}~g}}}
      } \\
      \equiv~&\StrP{\BNu{Y:\CTy}{B}}~A~m~\VThunk{\MTmLam{x}{\StrP{\BNu{Y:\CTy}{B}}~A'~\VThunk{\MForce{f}~x}~g}}
    \end{align*}
  \item \textbf{Case} $\StrP{X}$ for $X: \CTy$: since the structure is
    from the value environment, it's lawful by the induction hypothesis.
  \item \textbf{Case} $\StrP{(\SLam{X:K}{B})~S}$ for $S: K$ and $B: \CTy$: by induction hypothesis on $S$ and $B$
    \begin{align*}
      &\StrP{(\SLam{X:K}{B})~S}~A'~\VThunk{\MForce{m~\code{.bind}}~A~A'~m~f}~g \\
      \equiv~&(\MTyLam{X}{\MTmLam{\code{str}_{\code{X}}}{\StrP{B}}})~\floor{S}~\VThunk{\StrP{S}}~A'~\VThunk{\MForce{m~\code{.bind}}~A~A'~m~f}~g \\
      \equiv~&\subst{\subst{\StrP{B}}{\floor{S}}{X}}{\VThunk{\StrP{S}}}{\code{str}_{\code{X}}}~A'~\VThunk{\MForce{m~\code{.bind}}~A~A'~m~f}~g \\
      \equiv~&\subst{\subst{\StrP{B}}{\floor{S}}{X}}{\VThunk{\StrP{S}}}{\code{str}_{\code{X}}}~A~m~\VThunk{\MTmLam{x}{\subst{\subst{\StrP{B}}{\floor{S}}{X}}{\VThunk{\StrP{S}}}{\code{str}_{\code{X}}}~A'~\VThunk{\MForce{f}~x}~g}} \\
      \equiv~&\subst{\subst{\StrP{B}}{\floor{S}}{X}}{\VThunk{\StrP{S}}}{\code{str}_{\code{X}}}~A~m~\VThunk{\MTmLam{x}{\StrP{(\SLam{X:K}{B})~S}~A'~\VThunk{\MForce{f}~x}~g}} \\
      \equiv~&\StrP{(\SLam{X:K}{B})~S}~A~m~\VThunk{\MTmLam{x}{\StrP{(\SLam{X:K}{B})~S}~A'~\VThunk{\MForce{f}~x}~g}}
    \end{align*}
\end{enumerate}

Finally, we prove the linear bind law:
\begin{enumerate}
  \item \textbf{Case} $\StrP{\BReturn{A_0}}$ for all $A_0: \VTy$: by applying the linear bind law of monad $T$
    \begin{align*}
      &\MBind{m}{\MForce{tm}}{\StrP{\BReturn{A_0}}~A~m} \\
      \equiv~&\MBind{m}{\MForce{tm}}{\MForce{m~\code{.bind}}~A~A_0~m} \\
      \equiv~&\MForce{m~\code{.bind}}~A~A_0~\VThunk{\MBind{m}{\MForce{tm}}{\MForce{m}}} \\
      \equiv~&\StrP{\BReturn{A_0}}~A~\VThunk{\MBind{m}{\MForce{tm}}{\MForce{m}}}
    \end{align*}
  \item \textbf{Case} $\StrP{A_0 \to B}$ for $A_0: \VTy$ and $B: \CTy$: by induction hypothesis on $B$
    \begin{align*}
      &\MBind{m}{\MForce{tm}}{\StrP{A_0 \to B}~A~m} \\
      \equiv~&\MBind{m}{\MForce{tm}}{\MTmLam{f~a}{\StrP{B}~A~m~\VThunk{\MTmLam{x}{\MForce{f}~x~a}}}} \\
      \equiv~&\MTmLam{f~a}{\MBind{m}{\MForce{tm}}{\StrP{B}~A~m}}~\VThunk{\MTmLam{x}{\MForce{f}~x~a}} \\
      \equiv~&\MTmLam{f~a}{\StrP{B}~A~\VThunk{\MBind{m}{\MForce{tm}}{\MForce{m}}}}~\VThunk{\MTmLam{x}{\MForce{f}~x~a}} \\
      \equiv~&\StrP{B}~A~\VThunk{\MBind{m}{\MForce{tm}}{\MForce{m}}}
    \end{align*}
  \item \textbf{Case} $\StrP{\BWith{B_d}{d \in D}}$ for $\multi{B_d: \CTy}_{d \in D}$: by induction hypothesis on $\multi{B_d}_{d \in D}$
    \begin{align*}
      &\MBind{m}{\MForce{tm}}{\StrP{\BWith{B_d}{d \in D}}~A~m} \\
      \equiv~&\MBind{m}{\MForce{tm}}\MTmLam{f}\MComatch{.d}
        {\StrP{B_d}~A~m~\VThunk{\MTmLam{x}{\MForce{f}~x~.d}}}
      {d \in D} \\
      \equiv~&\MTmLam{f}\MComatch{.d}
        {\MBind{m}{\MForce{tm}}{\StrP{B_d}~A~m~\VThunk{\MTmLam{x}{\MForce{f}~x~.d}}}}
      {d \in D} \\
      \equiv~&\MTmLam{f}\MComatch{.d}
        {\StrP{B_d}~A~\VThunk{\MBind{m}{\MForce{tm}}{\MForce{m}}}~\VThunk{\MTmLam{x}{\MForce{f}~x~.d}}}
      {d \in D} \\
      \equiv~&\StrP{\BWith{B_d}{d \in D}}~A~\VThunk{\MBind{m}{\MForce{tm}}{\MForce{m}}}
    \end{align*}
  \item \textbf{Case} $\StrP{\BForall{X:K}{B}}$ for all $K$ and $B: \CTy$: by induction hypothesis on $B$
    \begin{align*}
      &\MBind{m}{\MForce{tm}}{\StrP{\BForall{X:K}{B}}~A~m} \\
      \equiv~&\MBind{m}{\MForce{tm}}{\MTmLam{f}\MTyLam{X}{\MTmLam{\code{str}_{\code{X}}}{\StrP{B}~A~m~\VThunk{\MTmLam{x}{\MForce{f}~x~X~\code{str}_{\code{X}}}}}}} \\
      \equiv~&\MTmLam{f}\MTyLam{X}{\MTmLam{\code{str}_{\code{X}}}{\MBind{m}{\MForce{tm}}{\StrP{B}~A~m}}}~\VThunk{\MTmLam{x}{\MForce{f}~x~X~\code{str}_{\code{X}}}} \\
      \equiv~&\MTmLam{f}\MTyLam{X}{\MTmLam{\code{str}_{\code{X}}}{\StrP{B}~A~\VThunk{\MBind{m}{\MForce{tm}}{\MForce{m}}}}~\VThunk{\MTmLam{x}{\MForce{f}~x~X~\code{str}_{\code{X}}}}} \\
      \equiv~&\StrP{B}~A~\VThunk{\MBind{m}{\MForce{tm}}{\MForce{m}}}
    \end{align*}
  \item \textbf{Case} $\StrP{\BNu{Y:\CTy}{B}}$ for $B: \CTy$: \newline
    let $B_\nu$ be $\BNu{Y:\CTy}{\floor{B}}$ \newline
    and $str_\nu$ be $\VThunk{\StrP{\BNu{Y:\CTy}{B}}}$ \newline
    by induction hypothesis on $B$
    \begin{align*}
      &\MBind{m}{\MForce{tm}}{\StrP{\BNu{Y:\CTy}{B}}~A~m} \\
      \equiv~&\MBind{m}{\MForce{tm}}{\MTmLam{f}\MRoll{
        \subst
          {\subst{\StrP{B}}{B_\nu}{Y}}
          {str_\nu}{str_Y}
          ~A~m~\VThunk{\MTmLam{x}{\MUnroll{\MForce{f}~x~}}}
      }} \\
      \equiv~&\MTmLam{f}\MRoll{\MBind{m}{\MForce{tm}}{
        \subst
          {\subst{\StrP{B}}{B_\nu}{Y}}
          {str_\nu}{str_Y}
          ~A~m
          ~\VThunk{\MTmLam{x}{\MUnroll{\MForce{f}~x~}}}
      }} \\
      \equiv~&\MTmLam{f}\MRoll{
        \subst
          {\subst{\StrP{B}}{B_\nu}{Y}}
          {str_\nu}{str_Y}
          ~A~\VThunk{\MBind{m}{\MForce{tm}}{\MForce{m}}}
          ~\VThunk{\MTmLam{x}{\MUnroll{\MForce{f}~x~}}}
      } \\
      \equiv~&\StrP{\BNu{Y:\CTy}{B}}~A~\VThunk{\MBind{m}{\MForce{tm}}{\MForce{m}}}
    \end{align*}
  \item \textbf{Case} $\StrP{X}$ for $X: \CTy$: since the structure is
    from the value environment, it's lawful by the induction hypothesis.
  \item \textbf{Case} $\StrP{(\SLam{X:K}{B})~S}$ for $S: K$ and $B: \CTy$: by induction hypothesis on $S$ and $B$
    \begin{align*}
      &\MBind{m}{\MForce{tm}}{\StrP{(\SLam{X:K}{B})~S}~A~m} \\
      \equiv~&\MBind{m}{\MForce{tm}}{\subst{\subst{\StrP{B}}{\floor{S}}{X}}{\VThunk{\StrP{S}}}{\code{str}_{\code{X}}}~A~m} \\
      \equiv~&\subst{\subst{\StrP{B}}{\floor{S}}{X}}{\VThunk{\StrP{S}}}{\code{str}_{\code{X}}}~A~\VThunk{\MBind{m}{\MForce{tm}}{\MForce{m}}} \\
      \equiv~&\StrP{(\SLam{X:K}{B})~S}~A~\VThunk{\MBind{m}{\MForce{tm}}{\MForce{m}}}
    \end{align*}
\end{enumerate}
\end{proof}

\fi % extended

\end{document}